{\global\def\bigPage{
        \setlength{\topmargin}{-0.5in}
        \setlength{\textheight}{9in}
        \setlength{\oddsidemargin}{+0.0in}
        \setlength{\textwidth}{6.3in}
        }
}

\documentclass[11pt]{article}\bigPage
\usepackage{amssymb,amsmath,amsthm,graphicx,psfrag}
\usepackage[usenames]{color}
\usepackage{url}
\usepackage[normalem]{ulem} 
\usepackage{hyperref}
\usepackage{enumitem} 
\usepackage{multirow} 

\usepackage{ifthen} 
\newif\ifarxiv
\arxivtrue 

\newcommand{\putFrag}[4]{\begin{figure}[htbp]
                            \begin{center}
                            #4
                            \includegraphics[width=#3in]{figures/#1.eps}
                            \end{center}
			    \caption{#2}
			    \label{fig:#1}
                          \end{figure} }

\newcommand{\textb}[1]{\textcolor{black}{#1}}

\newcommand{\blue}{\color{black}}
\newcommand{\black}{\color{black}}

 \newcommand{\defn}{\triangleq}

 \renewcommand{\vec}[1]{\ensuremath{\boldsymbol{#1}}}

 \newcommand{\mc}[1]{\ensuremath{\mathcal{#1}}}

 \newcommand{\Real}{{\mathbb{R}}}

 \newcommand{\of}[1]{^{\scriptscriptstyle (#1)}}
 \newcommand{\ofp}[1]{^{{\scriptscriptstyle (#1)}\prime}}
 \newcommand{\ofpp}[1]{^{{\scriptscriptstyle (#1)}\prime\prime}}

 \newcommand{\tran}{^{\top}}

 \DeclareMathOperator{\sgn}{sgn}
 
 \DeclareMathOperator{\E}{\mathbb{E}}

 \DeclareMathOperator*{\argmin}{arg\,min}

 \newtheorem{lemma}{Lemma}

 \newtheorem{assumption}{Assumption}
 \newtheorem{definition}{Definition}

 \renewcommand{\eqref}[1]{(\ref{eq:#1})}
 
 \newcommand{\Figref}[1]{Figure~\ref{fig:#1}}
 \newcommand{\figref}[1]{Fig.~\ref{fig:#1}}
 \newcommand{\tabref}[1]{Table~\ref{tab:#1}}
 \newcommand{\secref}[1]{Section~\ref{sec:#1}}
 
 \newcommand{\lemref}[1]{Lemma~\ref{lem:#1}}

 \newcommand{\assref}[1]{Assumption~\ref{ass:#1}}

 \newcommand{\denoiser}{\eta}
 \newcommand{\mse}{\mathcal{E}}
 \newcommand{\normal}{\mathcal{N}}
 \newcommand{\Amat}{\boldsymbol{\vec{A}}}

\begin{document}
\setlength{\arraycolsep}{0.8mm}

 \title{A Simple Derivation of AMP and its State Evolution via First-Order Cancellation}
 \author{Philip Schniter\thanks{P.~Schniter (email: schniter.1@osu.edu) is with
           the Department of Electrical and Computer Engineering,
           The Ohio State University, 
           Columbus, OH, 43210.
           His work is supported in part by
           the National Science Foundation under grant
           CCF-1716388.}}
 \date{}
 \maketitle

\begin{abstract}
We consider the linear regression problem, where the goal is to recover the vector $\boldsymbol{x}\in\mathbb{R}^n$ from measurements $\boldsymbol{y}=\boldsymbol{A}\boldsymbol{x}+\boldsymbol{w}\in\mathbb{R}^m$ under known matrix $\boldsymbol{A}$ and unknown noise $\boldsymbol{w}$. 
For large i.i.d.\ sub-Gaussian $\boldsymbol{A}$, the approximate message passing (AMP) algorithm is precisely analyzable through a state-evolution (SE) formalism, which furthermore shows that AMP is Bayes optimal in certain regimes.
The rigorous SE proof, however, is long and complicated.
And, although the AMP algorithm can be derived as an approximation of loop belief propagation (LBP), this viewpoint provides little insight into why large i.i.d.\ $\boldsymbol{A}$ matrices are important for AMP, and why AMP has a state evolution.
In this work, we provide a heuristic derivation of AMP and its state evolution, based on the idea of ``first-order cancellation,'' that provides insights missing from the LBP derivation while being much shorter than the rigorous SE proof.
\end{abstract}


\section{Introduction} \label{sec:intro}

We consider the standard linear regression problem, where the goal is to recover the vector $\vec{x}\in\Real^n$ from measurements 
\begin{align}
\vec{y}=\Amat\vec{x}+\vec{w}\in\Real^m
\label{eq:y},
\end{align}
where $\Amat$ is a known matrix and $\vec{w}$ is an unknown disturbance.
With high-dimensional random $\Amat$, the approximate message passing (AMP) algorithm \cite{Donoho:PNAS:09} remains one of the most celebrated and best understood iterative algorithms.
In particular, when the entries of $\Amat$ are drawn i.i.d.\ from a sub-Gaussian distribution and $m,n\rightarrow\infty$ with $m/n\rightarrow \delta\in(0,\infty)$, ensemble behaviors of AMP, such as the per-iteration mean-squared error (MSE), can be perfectly predicted using a state evolution (SE) formalism \cite{Bayati:AAP:15}.\footnote{See also \cite{Bayati:TIT:11} for an earlier proof of AMP's state evolution under i.i.d.\ Gaussian entries.}
Furthermore, the SE formalism shows that, in certain regimes, AMP's MSE converges to the minimum MSE as predicted by the replica method \cite{Bayati:TIT:11,Bayati:AAP:15}, which has been shown to coincide with the minimum MSE for linear regression under i.i.d.\ Gaussian $\Amat$ \cite{Reeves:ISIT:16,Barbier:ALL:16} as $m,n\rightarrow\infty$ with $m/n\rightarrow \delta\in(0,\infty)$.
More recently, it has been proven that the state-evolution accurately characterizes AMP's behavior for large but finite $m,n$ \cite{Rush:TIT:18}.

The rigorous SE proofs in \cite{Bayati:AAP:15,Bayati:TIT:11,Rush:TIT:18}, however, are long and complicated, and thus remain out of reach for many readers.
And, although the AMP algorithm can be heuristically derived from an approximation of loop belief propagation (LBP) \textb{\cite{Yedidia:TIT:05,Wainwright:FTML:08}}, as outlined in \cite{Donoho:ITW:10a}, \textb{\cite[App.A]{Bayati:TIT:11},} and \cite{Montanari:Chap:12}, \textb{and expectation propagation (EP) \cite{Minka:Diss:01,Heskes:JSM:05}, as outlined in \cite{Meng:SPL:15}},
the LBP\textb{/EP} perspective is lacking in several respects.
First, LBP \textb{and EP are heuristic, making it surprising that further approximations of these approaches can be optimal.}
Second, the LBP \textb{and EP derivations provide} little insight into why large i.i.d.\ $\Amat$ matrices are important for AMP.
\textb{For the LBP/EP derivations, it suffices that the entries of $\vec{A}$ have ``roughly the same magnitude $O(1/\sqrt{m})$'' \cite{Montanari:Chap:12,Meng:SPL:15}, suggesting that structured matrices (e.g., DCT, Hadamard, Fourier) should work as well as i.i.d.\ random ones. 
But, in practice, AMP often diverges with such structured matrices, and the reasons why are not evident from the LBP/EP viewpoint.}
Third, the LBP \textb{and EP derivations do} not \textb{explain why AMP obeys} a scalar state evolution \textb{when $\vec{A}$ is large and i.i.d.\ sub-Gaussian, nor do they describe how the estimation error variance can be predicted at each iteration}.

In this work, we propose a heuristic derivation of AMP and its MSE state evolution that uses the simple idea of ``first-order cancellation.'' 
This derivation provides insights missing from the LBP \text{and EP} derivations, while being much more accessible than the rigorous SE proofs.


\section{Problem Setup} \label{sec:setup}

In our treatment of the linear regression problem \eqref{y},
$\vec{y}=[y_1,\dots,y_m]\tran$, $\vec{x}=[x_1,\dots,x_n]\tran$, and $\vec{w}=[w_1,\dots,w_m]\tran$ are deterministic vectors and $\Amat\in\Real^{m\times n}$ is a deterministic matrix.
Importantly, however, we assume that the components $\{a_{ij}\}$ of $\Amat$ are realizations of 
i.i.d.\ Bernoulli\footnote{With additional work, our derivation can be extended to i.i.d.\ Gaussian $A_{ij}$, but doing so lengthens the derivation and provides little additional insight.} 
random variables $A_{ij}\in\pm\frac{1}{\sqrt{m}}$ that are drawn independently of $\vec{x}$ and $\vec{w}$.
Our model for $\Amat$ is a special case of that considered in \cite{Bayati:AAP:15}.

Throughout, we will focus on the following large-system limit.
\begin{definition}
The ``large system limit'' is defined as $m,n\rightarrow\infty$ with $m/n\rightarrow\delta$ for some fixed sampling ratio $\delta\in(0,\infty)$.
\end{definition}
\noindent
We will assume that the components of $\vec{x}$, $\vec{w}$, and $\vec{y}$ scale as $O(1)$ in the large-system limit.

We consider a family of algorithms that, starting with $\vec{x}\of{0}=\vec{0}$, iterates the following over iteration index $t=0,1,2,\dots$:
\begin{subequations} \label{eq:alg}
\begin{align}
\vec{v}\of{t} 
&= \vec{y} - \Amat\vec{x}\of{t} + \vec{\mu}\of{t}\\
\vec{x}\of{t+1} 
&= \denoiser\of{t}( \underbrace{\vec{x}\of{t}+\Amat\tran\vec{v}\of{t}}_{\displaystyle \defn \vec{r}\of{t}} ) ,
\end{align}
\end{subequations}
where $\vec{\mu}\of{t}$ is a correction term and
$\denoiser\of{t}(\cdot)$ is a component-wise \textb{separable} function.
\textb{That is, $[\denoiser\of{t}(\vec{r})]_j=\denoiser\of{t}(r_j)~\forall j$, where, with some abuse of notation, we use the same notation for the function $\eta\of{t}:\Real^n\rightarrow\Real^n$ and its component maps $\eta\of{t}:\Real\rightarrow\Real$.}
The quantity $\vec{x}\of{t}$ is iteration-$t$ estimate of the unknown vector $\vec{x}$.
We refer to $\denoiser\of{t}(\cdot)$ as a ``denoiser'' for reasons that will become clear in the sequel.
For technical reasons, we will assume that $\denoiser\of{t}(\cdot)$ is a polynomial function of bounded degree, similar to the assumption in \cite{Bayati:AAP:15}.

The classical iterative shrinkage/thresholding (IST) algorithm \cite{Chambolle:TIP:98} uses  no correction, i.e.,
\begin{align}
\vec{\mu}\of{t}=\vec{0}
\label{eq:ist} ,
\end{align}
for all iterations $t$,
whereas the AMP algorithm \cite{Donoho:PNAS:09} uses the ``Onsager'' correction 
\begin{align}
\vec{\mu}\of{t}
&= \frac{1}{m} \vec{v}\of{t-1} \sum_{j=1}^n \denoiser\ofp{t-1}(r_j\of{t-1})
\label{eq:onsager},
\end{align}
initialized with $\vec{\mu}\of{0}=\vec{0}$.
In \eqref{onsager}, $\denoiser\ofp{t}$ refers to the derivative of $\denoiser\of{t}$.
Our goal is to analyze the effect of $\vec{\mu}\of{t}$ on the behavior of algorithm \eqref{alg} in the large-system limit, and in particular to understand how and why the Onsager correction \eqref{onsager} is a good choice.
To do this, we will analyze the errors on $\vec{r}\of{t}$ and $\vec{x}\of{t}$ in \eqref{alg} and drop terms that vanish in the large-system limit.

\textb{It has been shown \cite{Bayati:TIT:11,Montanari:Chap:12} that IST has an predictable and desirable behavior in the case that $\vec{A}$ is a large i.i.d.\ Gaussian matrix that is \emph{re-drawn} at each iteration $t$ (with a corresponding update of $\vec{y}$).  
But this desirable behavior vanishes in the practical case that $\vec{A}$ is fixed over the iterations. 
In some sense, the goal of the Onsager correction \eqref{onsager} is to restore this desirable behavior when $\vec{A}$ is \emph{fixed} over the iterations.}


\section{AMP Derivation} \label{sec:deriv}

We will now analyze the error $\vec{e}\of{t}$ on the input to the denoiser $\vec{r}\of{t}$, i.e.,
\begin{align}
\vec{e}\of{t} 
\defn \vec{r}\of{t} - \vec{x} 
\label{eq:et0}.
\end{align}
From \eqref{alg} and \eqref{et0} we have that
\begin{align}
\vec{e}\of{t} 
&= \vec{x}\of{t} + \Amat\tran(\vec{y}-\Amat\vec{x}\of{t} + \vec{\mu}\of{t}) - \vec{x} \\
&= (\vec{I}-\Amat\tran\Amat)\vec{x}\of{t} + \Amat\tran(\Amat\vec{x}+\vec{w}+\vec{\mu}\of{t}) - \vec{x} \\
&= (\vec{I}-\Amat\tran\Amat)\vec{x}\of{t} - (\vec{I}-\Amat\tran\Amat)\vec{x} + \Amat\tran(\vec{w}+\vec{\mu}\of{t})  
\label{eq:et}.
\end{align}
Let us examine the $j$th component of $\vec{e}\of{t}$ when $t\geq 1$.
We have that
\ifarxiv{
\begin{align}
[ (\vec{I}-\Amat\tran\Amat)\vec{x}\of{t} ]_j
&= x\of{t}_j - \sum_i a_{ij} \sum_{l} a_{il} x_l\of{t} \\
&= \Big(1-\sum_{i=1}^m a_{ij}^2\Big)x\of{t}_j - \sum_i a_{ij} \sum_{l\neq j} a_{il} x_l\of{t}\\
&= - \sum_i a_{ij} \sum_{l\neq j} a_{il} x_l\of{t}
\label{eq:IAAxt0}
\end{align}
}\else{
\begin{align}
\lefteqn{ [ (\vec{I}-\Amat\tran\Amat)\vec{x}\of{t} ]_j }\nonumber\\
&= x\of{t}_j - \sum_i a_{ij} \sum_{l} a_{il} x_l\of{t} \\
&= \Big(1-\sum_{i=1}^m a_{ij}^2\Big)x\of{t}_j - \sum_i a_{ij} \sum_{l\neq j} a_{il} x_l\of{t}\\
&= - \sum_i a_{ij} \sum_{l\neq j} a_{il} x_l\of{t}
\label{eq:IAAxt0}
\end{align}
}\fi
since $a_{ij}^2=1/m~\forall ij$.
Continuing,
\ifarxiv{
\begin{align}
[ (\vec{I}-\Amat\tran\Amat)\vec{x}\of{t} ]_j
&= -\sum_i a_{ij} \sum_{l\neq j} a_{il} \denoiser\of{t-1}(r_l\of{t-1}) \\
&= -\sum_i a_{ij} \sum_{l\neq j} a_{il} \denoiser\of{t-1}\Big(
\underbrace{ x_l\of{t-1}+\sum_{k\neq i}a_{kl}v_k\of{t-1} 
        }_{\displaystyle \defn r_{il}\of{t-1}}
+a_{il}v_i\of{t-1}\Big) 
\label{eq:IAAxtj},
\end{align}
\else{
\begin{align}
\lefteqn{ [ (\vec{I}-\Amat\tran\Amat)\vec{x}\of{t} ]_j }\nonumber\\
&= -\sum_i a_{ij} \sum_{l\neq j} a_{il} \denoiser\of{t-1}(r_l\of{t-1}) \\
&= -\sum_i a_{ij} \sum_{l\neq j} a_{il} \denoiser\of{t-1}\Big(
\underbrace{ x_l\of{t-1}+\sum_{k\neq i}a_{kl}v_k\of{t-1} 
        }_{\displaystyle \defn r_{il}\of{t-1}}
+a_{il}v_i\of{t-1}\Big) 
\label{eq:IAAxtj},
\end{align}
}\fi
where $r_{il}\of{t-1}$ omits the direct contribution of $a_{il}$ from $r_l\of{t-1}$ and thus is only weakly dependent on $\{a_{ij}\}_{j=1}^n$.
We formalize this weak dependence through \assref{indep}, which is admittedly an approximation. 
In fact, the approximate nature of \assref{indep} is the main reason that our derivation is heuristic.
\begin{assumption}\label{ass:indep}
The matrix entry $a_{ij}$ is a realization of an equiprobable Bernoulli random variable $A_{ij}\in\pm\frac{1}{\sqrt{m}}$, where $\{A_{ij}\}$ are mutually independent and, $A_{ij}$ is independent of $\{r_{il}\of{t-1}\}_{l=1}^n$, $\{x_l\}_{l=1}^n$, and $\{w_k\}_{k=1}^m$.
\end{assumption}

\textb{
We now say a few words about \assref{indep}. 
The assumption that $\{A_{ij}\}$ are i.i.d.\ and independent of $\vec{x}$ and $\vec{w}$ is rather common in the compressive sensing literature.
For example, these assumptions are used in the rigorous AMP analyses \cite{Bayati:AAP:15,Bayati:TIT:11,Rush:TIT:18}.
The assumption of Bernoulli $A_{ij}$ is a bit stronger, but it is not critical, in that our analysis could be extended to handle other sub-Gaussian distributions on $A_{ij}$ with additional steps.
Doing so, however, would complicate the derivation without providing much additional insight, and so we have elected not to take this path.
The assumption that $A_{ij}$ is independent of $\{r_{il}\of{t-1}\}_{l=1}^n$ is far stronger.
In reality, there is a weak dependence between these quantities, but properly accounting for it seems to require completely different analysis methods, such as those in \cite{Bayati:AAP:15,Bayati:TIT:11,Rush:TIT:18}.
}

\assref{indep} will be used often when analyzing summations, as in the following lemma.
\begin{lemma} 
\label{lem:scaling}
Consider the quantity $z_i=\sum_{j=1}^n a_{ij} u_j$, where $a_{ij}$ are realizations of i.i.d.\ random variables $A_{ij}$ with zero mean and $\E[A_{ij}^2]=1/m$.
If $\{A_{ij}\}$ are drawn independently of $\{u_j\}$, and $\{u_j\}$ scale as $O(1)$ in the large-system limit, then $z_i$ also scales as $O(1)$.
\end{lemma}
\begin{proof}
First, note that $z_i$ is a realization of the random variable $Z_i\defn \sum_{j=1}^n A_{ij} u_j$.
Furthermore,
$\E[Z_i^2]
 = \E[(\sum_{j=1}^n A_{ij} u_j)^2]
 = \sum_{j=1}^n \sum_{l=1}^n \E[A_{ij}A_{il}] u_j u_l
 = \frac{1}{m} \sum_{j=1}^n u_j^2
 = \textb{\frac{n}{m}} \frac{1}{n}\sum_{j=1}^n u_j^2$,
since $\E[A_{ij}A_{il}] = 1/m$ if $j=l$ and 
$\E[A_{ij}A_{il}] =\E[A_{ij}]\E[A_{il}]= 0$ if $j\neq l$. 
Clearly $m/n$ and $\frac{1}{n}\sum_{j=1}^n u_j^2$ are both $O(1)$ 
in the large-system limit.
Thus we conclude that $\E[Z_i^2]$ is $O(1)$.
Finally, since $z_i$ is a realization of a random variable $Z_i$ whose second moment is $O(1)$, we conclude that $z_i$ scales as $O(1)$ in the large-system limit.
\end{proof}

\blue
We now emphasize the importance of the independence properties on $\{A_{ij}\}$ used in \assref{indep} and \lemref{scaling}.
For example, consider the case where $A_{ij}=\sgn(u_j)/\sqrt{m}$, so that $\E[A_{ij}^2]=1/m$ (as in \lemref{scaling}) but where $A_{ij}$ depends on $\{u_j\}$.
In this case, $z_i=\sum_{j=1}^n a_{ij} u_j = \frac{1}{\sqrt{m}}\sum_{j=1}^n |u_j|$, which scales as $O(\sqrt{m})$ in the large-system limit, and thus is fundamentally different from the $O(1)$ scaling observed in \lemref{scaling}.

The distinction between 
$\{a_{ij}\}$ being independent of other quantities,
versus $\{a_{ij}\}$ simply being the right size, is one of the major
differences between our derivation 
and other simple derivations based on LBP and EP.
For example, both \cite{Montanari:Chap:12} and \cite{Meng:SPL:15} require only that ``$\{a_{ij}\}$ are roughly the same magnitude $O(1/\sqrt{m})$,'' 
which suggests that properly normalized structured matrices (e.g, Hadamard, Fourier, Discrete Cosine Transform, Radon Transform) should work well with AMP.
But when used to recover natural images/signals with these matrices, AMP often diverges.
And the LBP/EP viewpoint does not explain why.
\black

In the sequel, we will make use of the following lemma, whose proof is postponed because it is a bit long and does not provide much insight.
\begin{lemma}
\label{lem:scaling2}
Under \assref{indep} and the Onsager choice of $\vec{\mu}\of{t}$ from \eqref{onsager}, the elements of $\vec{v}\of{t}$, $\vec{r}\of{t}$, $\vec{x}\of{t}$, and $\vec{\mu}\of{t}$ scale as $O(1)$ in the large-system limit for all iterations $t$.
\end{lemma}
\begin{proof}
See the appendix.
\end{proof}

We now perform a Taylor series expansion of the $\denoiser\of{t-1}$ term in \eqref{IAAxtj} about $r_{il}\of{t-1}$:
\ifarxiv{
\begin{align}
\lefteqn{ 
\denoiser\of{t-1}(r_{il}\of{t-1} + a_{il}v_i\of{t-1}) }\nonumber\\
&= \denoiser\of{t-1}(r_{il}\of{t-1})
  + a_{il}v_i\of{t-1} \denoiser\ofp{t-1}(r_{il}\of{t-1}) 
  + \underbrace{ \frac{1}{2}a_{il}^2 (v_i\of{t-1})^2 \denoiser\ofpp{t-1}(r_{il}\of{t-1}) 
                 + \text{H.O.T.} }_{\displaystyle O(1/m)} 
\label{eq:Taylor1} ,
\end{align}
}\else{
\begin{align}
\lefteqn{ 
\denoiser\of{t-1}(r_{il}\of{t-1} + a_{il}v_i\of{t-1}) }\nonumber\\
&= \denoiser\of{t-1}(r_{il}\of{t-1})
  + a_{il}v_i\of{t-1} \denoiser\ofp{t-1}(r_{il}\of{t-1}) 
  \nonumber\\&\quad
  + \underbrace{ \frac{1}{2}a_{il}^2 (v_i\of{t-1})^2 \denoiser\ofpp{t-1}(r_{il}\of{t-1}) 
                 + \text{H.O.T.} }_{\displaystyle O(1/m)} 
\label{eq:Taylor1} ,
\end{align}
}\fi
where the $O(1/m)$ scaling follows from the fact 
that $a_{il}^2=1/m~\forall il$, 
that both $v_i\of{t-1}$ and $r_{il}\of{t-1}$ scale as $O(1)$ via \lemref{scaling},
and $\denoiser\of{t-1}(\cdot)$ is polynomial of bounded degree,
which implies that $\denoiser\ofpp{t-1}(r_{il}\of{t-1})$ also scales as $O(1)$.
\textb{For similar reasons, the second term in \eqref{Taylor1} scales as $O(1/\sqrt{m})$.}
We will ignore the $O(1/m)$ term in \eqref{Taylor1} since it vanishes relative to the \textb{$O(1/\sqrt{m})$} component in the large-system limit.
Thus we have
\ifarxiv{
\begin{align}
[ (\vec{I}-\Amat\tran\Amat)\vec{x}\of{t} ]_j
&\approx -\sum_i a_{ij} \sum_{l\neq j} a_{il} 
  \big[ \denoiser\of{t-1}(r_{il}\of{t-1}) + a_{il}v_i\of{t-1}\denoiser\ofp{t-1}(r_{il}\of{t-1}) 
        \big] \\
&= -\sum_i a_{ij} \sum_{l\neq j} a_{il} \denoiser\of{t-1}(r_{il}\of{t-1}) 
        - \frac{1}{m}\sum_i a_{ij} v_i\of{t-1} \sum_{l\neq j} \denoiser\ofp{t-1}(r_{il}\of{t-1}) 
\label{eq:IAAxt}
\end{align}
}\else{
\begin{align}
\lefteqn{ [ (\vec{I}-\Amat\tran\Amat)\vec{x}\of{t} ]_j }\nonumber\\
&\approx -\sum_i a_{ij} \sum_{l\neq j} a_{il} 
  \big[ \denoiser\of{t-1}(r_{il}\of{t-1}) + a_{il}v_i\of{t-1}\denoiser\ofp{t-1}(r_{il}\of{t-1}) 
        \big] \\
&= -\sum_i a_{ij} \sum_{l\neq j} a_{il} \denoiser\of{t-1}(r_{il}\of{t-1}) 
        \nonumber\\&\quad
        - \frac{1}{m}\sum_i a_{ij} v_i\of{t-1} \sum_{l\neq j} \denoiser\ofp{t-1}(r_{il}\of{t-1}) 
\label{eq:IAAxt}
\end{align}
}\fi
using $a_{il}^2=1/m~\forall il$.  

Similar to \textb{\eqref{IAAxt0}}, we have
\begin{align}
[ (\vec{I}-\Amat\tran\Amat)\vec{x} ]_j
&= - \sum_i a_{ij} \sum_{l\neq j} a_{il} x_l ,
\end{align}
which, combined with \eqref{et} and \eqref{IAAxt}, yields 
\ifarxiv{
\begin{align}
e_j\of{t}
&\approx \sum_i a_{ij} \sum_{l\neq j} a_{il} \big[ x_l - \denoiser\of{t-1}(r_{il}\of{t-1}) \big] 
   - \frac{1}{m} \sum_i a_{ij} v_i\of{t-1} \sum_{l\neq j} \denoiser\ofp{t-1}(r_{il}\of{t-1}) 
   + \sum_i a_{ij} (w_i+\mu_i\of{t}) \\
&= \sum_i a_{ij} \sum_{l\neq j} a_{il} \big[ x_l - \denoiser\of{t-1}(r_{il}\of{t-1}) \big] 
  + \sum_i a_{ij} w_i 
  + \sum_i a_{ij} \Big[ \mu_i\of{t}
        - v_i\of{t-1} \frac{1}{m}\sum_{l\neq j} \denoiser\ofp{t-1}(r_{il}\of{t-1}) \Big] .
\label{eq:ejt}
\end{align}
}\else{
\begin{align}
&e_j\of{t}
\approx \sum_i a_{ij} \sum_{l\neq j} a_{il} \big[ x_l - \denoiser\of{t-1}(r_{il}\of{t-1}) \big] 
  \nonumber\\&\quad
  - \frac{1}{m} \sum_i a_{ij} v_i\of{t-1} \sum_{l\neq j} \denoiser\ofp{t-1}(r_{il}\of{t-1}) 
  + \sum_i a_{ij} (w_i+\mu_i\of{t}) \\
&= \sum_i a_{ij} \sum_{l\neq j} a_{il} \big[ x_l - \denoiser\of{t-1}(r_{il}\of{t-1}) \big] 
  \nonumber\\&\quad
  + \sum_i a_{ij} w_i 
  + \sum_i a_{ij} \Big[ \mu_i\of{t}
        - v_i\of{t-1} \frac{1}{m}\sum_{l\neq j} \denoiser\ofp{t-1}(r_{il}\of{t-1}) \Big] .
\label{eq:ejt}
\end{align}
}\fi

We are now in a position to observe the principal mechanism of AMP.
As we argue below (using the central limit theorem), the first and second terms in \eqref{ejt} behave like realizations of zero-mean Gaussians in the large-system limit, because
$\{a_{il}\}$ are realizations of i.i.d.\ zero-mean random variables $\{A_{il}\}$
that are independent of $x_l$, $w_i$, and $\{r_{il}\of{t-1}\}$ under \assref{indep}.
But the same cannot be said in general about the third term in \eqref{ejt}, because $v_i\of{t-1}$ is strongly coupled to $a_{ij}$.
Consequently, the denoiser input-error $e_j\of{t}$ is difficult to characterize for general $\mu_i\of{t}$.

With AMP's choice of $\mu_i\of{t}$, however, the 3rd term in \eqref{ejt} vanishes in the large-system limit.
In particular, with the Onsager choice \eqref{onsager}, the 3rd term in \eqref{ejt} takes the form
\ifarxiv{
\begin{align}
\lefteqn{
\sum_i a_{ij} \Big[ \frac{1}{m} v_i\of{t-1}\sum_l \denoiser\ofp{t-1}(r_l\of{t-1})
-\frac{1}{m} v_i\of{t-1}\sum_{l\neq j} \denoiser\ofp{t-1}(r_{il}\of{t-1}) \Big]
}\nonumber\\
&= \frac{1}{m}\sum_i a_{ij} v_i\of{t-1} \Big[ \denoiser\ofp{t-1}(r_j\of{t-1}) +\sum_{l\neq j} \Big(\denoiser\ofp{t-1}(r_l\of{t-1})-\denoiser\ofp{t-1}(r_{il}\of{t-1}) \Big)\Big] \\
&\approx \frac{1}{m} \sum_i a_{ij} v_i\of{t-1} \Big[ \denoiser\ofp{t-1}(r_j\of{t-1}) + \sum_{l\neq j} a_{il}v_i\of{t-1}\denoiser\ofpp{t-1}(r_{il}\of{t-1}) \Big] 
\label{eq:bad},
\end{align}
}\else{
\begin{align}
\lefteqn{
\sum_i a_{ij} \Big[ \frac{v_i\of{t-1}}{m}\sum_l \denoiser\ofp{t-1}(r_l\of{t-1})
-\frac{v_i\of{t-1}}{m} \sum_{l\neq j} \denoiser\ofp{t-1}(r_{il}\of{t-1}) \Big]
}\nonumber\\
&= \frac{1}{m}\sum_i a_{ij} v_i\of{t-1} \Big[ \denoiser\ofp{t-1}(r_j\of{t-1}) 
  \nonumber\\&\quad
  +\sum_{l\neq j} \Big(\denoiser\ofp{t-1}(r_l\of{t-1})-\denoiser\ofp{t-1}(r_{il}\of{t-1}) \Big)\Big] \\
&\approx \frac{1}{m} \sum_i a_{ij} v_i\of{t-1} \Big[ \denoiser\ofp{t-1}(r_j\of{t-1}) 
  \nonumber\\&\quad
  + \sum_{l\neq j} a_{il}v_i\of{t-1}\denoiser\ofpp{t-1}(r_{il}\of{t-1}) \Big] 
\label{eq:bad}, \hspace{30mm}
\end{align}
}\fi
where for the last step we used the Taylor-series expansion
\ifarxiv{
\begin{align}
\denoiser\ofp{t-1}(r_{l}\of{t-1})
=\denoiser\ofp{t-1}(r_{il}\of{t-1}+a_{il}v_i\of{t-1})
&= \denoiser\ofp{t-1}(r_{il}\of{t-1}) + a_{il}v_i\of{t-1}\denoiser\ofpp{t-1}(r_{il}\of{t-1}) + O(1/m) 
\label{eq:Taylor2}
\end{align}
}\else{
\begin{align}
\lefteqn{\denoiser\ofp{t-1}(r_{l}\of{t-1})}\nonumber\\
&=\denoiser\ofp{t-1}(r_{il}\of{t-1}+a_{il}v_i\of{t-1}) \\
&= \denoiser\ofp{t-1}(r_{il}\of{t-1}) + a_{il}v_i\of{t-1}\denoiser\ofpp{t-1}(r_{il}\of{t-1}) + O(1/m) 
\label{eq:Taylor2}
\end{align}
}\fi
and dropped the $O(1/m)$ term, since it will vanish relative to the $a_{il}v_i\of{t-1}\denoiser\ofpp{t-1}(r_{il}\of{t-1})$ term in the large-system limit.
Looking at \eqref{bad}, the first term is
\begin{align}
\frac{1}{m} \sum_{i=1}^m \underbrace{ a_{ij} v_i\of{t-1} \denoiser\ofp{t-1}(r_{j}\of{t-1}) }_{\displaystyle O(1/\sqrt{m}) }
= O(1/\sqrt{m})
\end{align}
since $a_{ij}\in\pm 1/\sqrt{m}$ and $v_i\of{t-1} \denoiser\ofp{t-1}(r_{j}\of{t-1})$ is $O(1)$ due to \lemref{scaling}. 
Thus the first term in \eqref{bad} will vanish in the large-system limit.
The second term in \eqref{bad} is
\begin{align}
\frac{1}{m} \sum_{i=1}^m 
      \underbrace{
        a_{ij} (v_i\of{t-1})^2 \underbrace{ 
        \sum_{l\neq j} a_{il}\denoiser\ofpp{t-1}(r_{il}\of{t-1}) }_{\displaystyle O(1)}
      }_{\displaystyle O(1/\sqrt{m})}
= O(1/\sqrt{m}) 
\label{eq:bad2},
\end{align}
which will also vanish in the large-system limit.
The $O(1)$ scaling in \eqref{bad2} follows from \lemref{scaling} under \assref{indep},
and the $O(1/\sqrt{m})$ scaling follows from the fact that $a_{il}\in\pm 1/\sqrt{m}$ and $(v_i\of{t-1})^2=O(1)$.

Thus, for large $m$ and the AMP choice of $\mu_i\of{t}$, equation \eqref{ejt} becomes 
\begin{align}
e_j\of{t}
&\approx \sum_i a_{ij} \sum_{l\neq j} a_{il} \big[ 
        \underbrace{ x_l - \denoiser\of{t-1}(r_{il}\of{t-1}) 
                }_{\displaystyle \defn \epsilon_{il}\of{t}}
        \big] + \sum_i a_{ij} w_i 
\label{eq:ejt2} .
\end{align}
\textb{Recall that, under \assref{indep}, $a_{ij}$ is a realization of equiprobable $A_{ij}\in\pm\frac{1}{\sqrt{m}}$ that is independent of 
$\{A_{il}\}_{l\neq j}$, $\{x_l\}_{l=1}^n$, and $\{r_{il}\of{t-1}\}_{l=1}^n$. 
This implies that $A_{ij}$ is also independent of $\epsilon_{il}\of{t}=x_l - \denoiser\of{t-1}(r_{il}\of{t-1})$ for any $l$.}
Thus we can apply the central limit theorem to say that, for any fixed $\{\epsilon_{il}\of{t}\}$, the first term converges to a Gaussian with mean and variance
\ifarxiv{
\begin{align}
\E\bigg[ \sum_i A_{ij} \sum_{l\neq j} A_{il} \epsilon_{il}\of{t} \bigg] 
&= \sum_i \E[A_{ij}] \sum_{l\neq j} \E[ A_{il} ] \epsilon_{il}\of{t} = 0 \\
\E\Big[ \Big( \sum_i A_{ij} \sum_{l\neq j} A_{il} \epsilon_{il}\of{t} \Big)^2 \Big]
&= \sum_i \E[A_{ij}^2] \sum_{l\neq j} \E[ A_{il}^2 ] (\epsilon_{il}\of{t})^2 
= \frac{1}{m^2} \sum_i \sum_{l\neq j} (\epsilon_{il}\of{t})^2 
\label{eq:variance} .
\end{align}
}\else{
\begin{align}
\E\Big[ \sum_i A_{ij} \sum_{l\neq j} A_{il} \epsilon_{il}\of{t} \Big] 
&= \sum_i \E[A_{ij}] \sum_{l\neq j} \E[ A_{il} ] \epsilon_{il}\of{t} \\
&= 0 
\end{align}
\begin{align}
\E\bigg[ \Big( \sum_i A_{ij} \sum_{l\neq j} A_{il} \epsilon_{il}\of{t} \Big)^2 \bigg]
&= \sum_i \E[A_{ij}^2] \sum_{l\neq j} \E[ A_{il}^2 ] (\epsilon_{il}\of{t})^2\\
&= \frac{1}{m^2} \sum_i \sum_{l\neq j} (\epsilon_{il}\of{t})^2 
\label{eq:variance} .
\end{align}
}\fi
From the Taylor expansion \eqref{Taylor1}, we have
\begin{align}
\epsilon_{il}\of{t}
&= x_l - \denoiser\of{t-1}(r_{il}\of{t-1}) \\
&= \underbrace{ x_l - \denoiser\of{t-1}(r_{l}\of{t-1}) 
        }_{\displaystyle \defn \epsilon_l\of{t}}
   + \underbrace{ a_{il}v_i\of{t-1} \denoiser\ofp{t-1}(r_{il}\of{t-1}) + O(1/m) 
        }_{\displaystyle O(1/\sqrt{m})} 
\label{eq:epsilt},
\end{align}
where the $O(1/\sqrt{m})$ scaling follows from the facts that $a_{il}\in\pm 1/\sqrt{m}$ and $v_i\of{t-1} \denoiser\ofp{t-1}(r_{il}\of{t-1})$ is $O(1)$.
Notice that $\epsilon_l\of{t}$ is the denoiser output error.
\textb{Because both $x_l$ and $r_l\of{t-1}$ are $O(1)$, it follows that $\epsilon_l\of{t}$ is also $O(1)$.}
Because the $O(1/\sqrt{m})$ term in \eqref{epsilt} vanishes in the large-system limit,
we see that \eqref{variance} becomes
\begin{align}
\frac{1}{m^2} \sum_i \sum_{l\neq j} (\epsilon_{il}\of{t})^2 
&\approx \frac{1}{m^2} \sum_{i=1}^m \sum_{l\neq j} (\epsilon_l\of{t})^2 
= \frac{1}{m} \sum_{l\neq j} (\epsilon_l\of{t})^2 \\
&= \underbrace{ \frac{n}{m}\frac{1}{n}\sum_{l=1}^n (\epsilon_l\of{t})^2 }_{\displaystyle O(1)} -\underbrace{ \frac{1}{m}(\epsilon_j\of{t})^2 }_{\displaystyle O(1/m)} 
\approx \delta^{-1} \mse\of{t} 
\label{eq:variance2} ,
\end{align}
where 
\begin{align}
\mse\of{t} \defn \lim_{n\rightarrow\infty} \frac{1}{n} \sum_{l=1}^n (\epsilon_l\of{t})^2 
\label{eq:mset}
\end{align}
is the average squared error on the denoiser output $\vec{x}\of{t}$.
We have thus deduced that, in the large-system limit, the first term in \eqref{ejt2} behaves like a zero-mean Gaussian with variance $\delta^{-1}\mse\of{t}$.
For the second term in \eqref{ejt2}, we can again use the central limit theorem to say that, for any fixed $\{w_i\}$, the second term converges to a Gaussian with mean and variance
\begin{align}
\E\Big[ \sum_i A_{ij} w_i \Big] 
&= \sum_i \E[A_{ij}] w_i = 0 \\
\E\bigg[ \Big( \sum_i A_{ij} w_i \Big)^2 \bigg]
&= \sum_i \E[A_{ij}^2] w_i^2 
= \frac{1}{m} \sum_{i=1}^m w_i^2 \approx \tau_w ,
\end{align}
where $\tau_w$ denotes the empirical second moment of the noise:
\begin{align}
\tau_w \defn \lim_{m\rightarrow \infty} \frac{1}{m} \sum_{i=1}^m w_i^2
\label{eq:tauw} .
\end{align}

To summarize, with AMP's choice of $\vec{\mu}\of{t}$ from \eqref{onsager}, the $j$th component of the denoiser input-error behaves like
\begin{align}
e_j\of{t} 
\sim \normal\big(0,\underbrace{ \delta^{-1} \mse\of{t} + \tau_w }_{\displaystyle \defn \tau_r\of{t}} \big)
\label{eq:taurt}
\end{align}
in the large-system limit, where $\normal(\mu,\sigma^2)$ denotes a Gaussian random variable with mean $\mu$ and variance $\sigma^2$.
With other choices of $\vec{\mu}\of{t}$ (e.g., IST's choice of $\vec{\mu}\of{t}=\vec{0}~\forall t$), it is difficult to characterize the denoiser input-error $\vec{e}\of{t}$ and in general it will not be Gaussian.


\section{AMP State Evolution} \label{sec:SE}

In \secref{deriv}, we used \assref{indep} to argue that the AMP algorithm yields a denoiser input-error $\vec{e}\of{t}$ whose components are $\normal(0,\tau_r\of{t})$ in the large system limit. 
Here, $\tau_r\of{t}=\delta^{-1}\mse\of{t}+\tau_w$ where 
$\mse\of{t}$ is the average squared-error at the denoiser output in the large-system limit.

Recalling the definition of $\mse\of{t}$ from \eqref{mset}, we can write
\begin{align}
\frac{1}{n} \sum_{l=1}^n (\epsilon_l\of{t})^2 
&\approx \frac{1}{n} \sum_{l=1}^n \big[ \eta\of{t-1}(x_l + \mc{N}(0,\tau_r\of{t-1})) -x_l \big]^2 \\
&= \E\big[\eta\of{t-1}(X+\mc{N}(0,\tau_r\of{t-1})) - X\big]^2
\end{align}
where $X$ is a scalar random variable defined from the empirical distribution
\begin{align}
X \sim p(x) = \frac{1}{n}\sum_{l=1}^n \delta(x-x_l) ,
\end{align}
with $\delta(\cdot)$ denoting the Dirac delta function.
Thus we can argue that, in the large-system limit,
\begin{align}
\mse\of{t} = \E\big[ \denoiser\of{t\textb{-1}}\big(X+\normal(0,\tau_r\of{t-1})\big) - X \big]^2 
\label{eq:mset1} ,
\end{align}
where $X$ now is distributed according to the $n\rightarrow\infty$ limit of the empirical distribution.
Combining \eqref{mset1} with the update equation for $\tau_r\of{t}$ gives the following recursion for $t=0,1,2,\dots$:%
\begin{subequations}\label{eq:se}
\begin{align}
\tau_r\of{t} 
&= \delta^{-1} \mse\of{t} + \tau_w \\
\mse\of{t+1} 
&= \E\big[ \denoiser\of{t}\big(X+\normal(0,\tau_r\of{t})\big) - X \big]^2 ,
\end{align}
\end{subequations}
initialized with $\mse\of{0}=\E[X^2]$.
The recursion \eqref{se} is known as AMP's ``state evolution'' for the mean-squared error \cite{Donoho:PNAS:09,Bayati:TIT:11,Bayati:AAP:15}.

The reason that we call $\denoiser\of{t}(\cdot)$ a ``denoiser'' should now be clear.
To minimize the mean-squared error $\mse\of{t+1}$, the function $\denoiser\of{t}(\cdot)$ should remove as much of the noise from its input as possible.
The smaller that $\mse\of{t+1}$ is, the smaller the input-noise variance $\tau_r\of{t+1}$ will be during the next iteration.


\section{AMP Variance Estimation} \label{sec:var}

For best performance, the iteration-$t$ denoiser $\denoiser\of{t}(\cdot)$ should be designed in accordance with the iteration-$t$ input noise variance $\tau_r\of{t}$.
With the AMP algorithm, there is an easy way to estimate the value of $\tau_r\of{t}$ at each iteration $t$ from the $\vec{v}\of{t}$ vector, i.e.,
$\tau_r\of{t}\approx \|\vec{v}\of{t}\|^2/m$ \cite{Montanari:Chap:12}. 
Below, we explain this approach using arguments similar to those used above.

Equation \eqref{vit} shows that 
\ifarxiv{
\begin{align}
v_i\of{t}
&= y_i - \sum_{l=1}^n a_{il} \denoiser\of{t-1}(r_{il}\of{t-1})
- v_i\of{t-1} \frac{1}{m}\sum_{l=1}^n \denoiser\ofp{t-1}(r_{il}\of{t-1})
+ \mu_i\of{t} + O(1/m) .
\end{align}
}\else{
\begin{align}
v_i\of{t}
&= y_i - \sum_{l=1}^n a_{il} \denoiser\of{t-1}(r_{il}\of{t-1})
\nonumber\\&\quad
- v_i\of{t-1} \frac{1}{m}\sum_{l=1}^n \denoiser\ofp{t-1}(r_{il}\of{t-1})
+ \mu_i\of{t} + O(1/m) .
\end{align}
}\fi
Ignoring the $O(1/m)$ term and plugging in AMP's choice of $\mu_i\of{t}$ from \eqref{onsager} yields
\ifarxiv{
\begin{align}
v_i\of{t}
&\approx y_i - \sum_{l=1}^n a_{il} \denoiser\of{t-1}(r_{il}\of{t-1})
+ v_i\of{t-1} \frac{1}{m}\sum_{l=1}^n \big[ \denoiser\ofp{t-1}(r_{l}\of{t-1})
  -\denoiser\ofp{t-1}(r_{il}\of{t-1}) \big] \\
&= y_i - \sum_{l=1}^n a_{il} \denoiser\of{t-1}(r_{il}\of{t-1})
+ v_i\of{t-1} \frac{n}{m} \frac{1}{n}\sum_{l=1}^n 
  \underbrace{
  \big[ a_{il} v_i\of{t-1}\denoiser\ofpp{t-1}(r_{il}\of{t-1}) + O(1/m) \big] 
  }_{\displaystyle O(1/\sqrt{m})} 
\label{eq:vit1},
\end{align}
}\else{
\begin{align}
v_i\of{t}
&\approx y_i - \sum_{l=1}^n a_{il} \denoiser\of{t-1}(r_{il}\of{t-1})
\nonumber\\&\quad
+ v_i\of{t-1} \frac{1}{m}\sum_{l=1}^n \big[ \denoiser\ofp{t-1}(r_{l}\of{t-1})
  -\denoiser\ofp{t-1}(r_{il}\of{t-1}) \big] \\
&= y_i - \sum_{l=1}^n a_{il} \denoiser\of{t-1}(r_{il}\of{t-1})
\nonumber\\&\quad
+ v_i\of{t-1} \frac{n}{m} \frac{1}{n}\sum_{l=1}^n 
  \underbrace{
  \big[ a_{il} v_i\of{t-1}\denoiser\ofpp{t-1}(r_{il}\of{t-1}) + O(1/m) \big] 
  }_{\displaystyle O(1/\sqrt{m})} 
\label{eq:vit1},
\end{align}
}\fi
where we used the Taylor series \eqref{Taylor2} in the second step and $a_{il}\in\pm 1/\sqrt{m}$ to justify the $O(1/\sqrt{m})$ scaling.
Since the last term in \eqref{vit1} is the scaled average of $O(1/\sqrt{m})$ terms, with $O(1)$ scaling, the entire term is $O(1/\sqrt{m})$.
We can thus drop it since it will vanish relative to the others in the large-system limit. 
Doing this and plugging in $\vec{y}=\vec{Ax}+\vec{w}$ yields 
\begin{align}
v_i\of{t}
&\approx w_i + \sum_{l=1}^n a_{il} \big[ \underbrace{ x_l-\denoiser\of{t-1}(r_{il}\of{t-1}) }_{\displaystyle =\epsilon_{il}\of{t}} \big] ,
\end{align}
recalling the definition of $\epsilon_{il}\of{t}$ from \eqref{ejt2}.
Squaring the result and averaging over $i$ yields
\ifarxiv{
\begin{align}
\frac{1}{m}\sum_{i=1}^m (v_i\of{t})^2
&\approx \frac{1}{m}\sum_{i=1}^m w_i^2 + \frac{1}{m}\sum_{i=1}^m \bigg( \sum_{l=1}^n a_{il} \epsilon_{il}\of{t} \bigg)^2 + \frac{2}{m}\sum_{i=1}^m \bigg( w_i \sum_{l=1}^n a_{il} \big[ x_l-\denoiser\of{t-1}(r_{il}\of{t-1}) \big] \bigg)
\label{eq:taurhat}.
\end{align}
}\else{
\begin{align}
\frac{1}{m}\sum_{i=1}^m (v_i\of{t})^2
&\approx \frac{1}{m}\sum_{i=1}^m w_i^2 + \frac{1}{m}\sum_{i=1}^m \bigg( \sum_{l=1}^n a_{il} \epsilon_{il}\of{t} \bigg)^2 
\nonumber\\&\quad
+ \frac{2}{m}\sum_{i=1}^m \bigg( w_i \sum_{l=1}^n a_{il} \big[ x_l-\denoiser\of{t-1}(r_{il}\of{t-1}) \big] \bigg)
\label{eq:taurhat}.
\end{align}
}\fi
We now examine the components of \eqref{taurhat} in the large-system limit.
By definition, the first term in \eqref{taurhat} converges to $\tau_w$.
By the law of large numbers, the second term converges to 
\ifarxiv{
\begin{align}
\lim_{n\rightarrow\infty} 
\E\bigg[\bigg( \sum_{l=1}^n A_{il} \epsilon_{il}\of{t} \bigg)^2\bigg]
&= \lim_{n\rightarrow\infty} 
\sum_{l=1}^n \sum_{j=1}^n \E[ A_{il} A_{ij} ] \epsilon_{il}\of{t} \epsilon_{ij}\of{t} 
= \lim_{n\rightarrow\infty} 
\frac{1}{m} \sum_{l=1}^n (\epsilon_{il}\of{t})^2 ,
\end{align}
}\else{
\begin{align}
\lim_{n\rightarrow\infty} 
\E\bigg[\bigg( \sum_{l=1}^n A_{il} \epsilon_{il}\of{t} \bigg)^2\bigg]
&= \lim_{n\rightarrow\infty} 
\sum_{l=1}^n \sum_{j=1}^n \E[ A_{il} A_{ij} ] \epsilon_{il}\of{t} \epsilon_{ij}\of{t} \\
&= \lim_{n\rightarrow\infty} 
\frac{1}{m} \sum_{l=1}^n (\epsilon_{il}\of{t})^2 ,
\end{align}
}\fi
since $\E[ A_{il} A_{ij} ]=1/m$ when $l=j$ and $\E[ A_{il} A_{ij} ]=0$ when $l\neq j$.
Using the relationship between $\epsilon_{il}\of{t}$ and $\epsilon_l\of{t}$ from \eqref{epsilt}, it can be seen that
\begin{align}
\lim_{n\rightarrow\infty} 
\frac{1}{m} \sum_{l=1}^n (\epsilon_{il}\of{t})^2 
&= \lim_{n\rightarrow\infty} \frac{n}{m} \frac{1}{n} \sum_{l=1}^n (\epsilon_{l}\of{t})^2 
= \delta^{-1} \mc{E}\of{t}  
\end{align}
where $m$ is implicitly a function of $n$ because $m/n=O(1)$. 
In summary,
\begin{align}
\lim_{m\rightarrow\infty} \frac{1}{m}\sum_{i=1}^m (v_i\of{t})^2
= \tau_w + \delta^{-1} \mc{E}\of{t} = \tau_r\of{t},
\end{align}
which shows that $\tau_r\of{t}$ is well estimated by $\|\vec{v}\of{t}\|^2/m$ in the large-system limit.


\newcommand{\sz}{0.6}

\section{Numerical Experiments}

We now present numerical experiments that demonstrate the AMP behaviors discussed above.
In all experiments, we used
a sampling ratio of $\delta=0.5$,
$\{A_{ij}\}$ drawn i.i.d.\ zero-mean Gaussian with variance $1/m$,
$\{x_j\}$ drawn i.i.d.\ from the Bernoulli-Gaussian distribution with sparsity rate $\beta=0.1$ (i.e.,
$p_X(x_j)=(1-\beta)\delta(x_j)+\beta\mc{N}(x_j;0,1)~\forall j$, where $\delta(\cdot)$ denotes the Dirac delta distribution),
and $\{w_{i}\}$ drawn i.i.d.\ zero-mean Gaussian with variance $\beta 10^{-\text{SNRdB}/10}$ and $\text{SNRdB}=20$, so that $\E[\|\vec{Ax}\|^2]/\E[\|\vec{w}\|^2]\approx 20$ dB.
We experimented with two denoisers: the MMSE denoiser $\denoiser\of{t}(r_j)=\E[X\,|\,r_j=X+\mc{N}(0,\tau_r\of{t})]$ and the soft-thresholding denoiser $\denoiser\of{t}(r_j)=\sgn(r_j)\max\{0,|r_j|-\alpha\sqrt{\tau_r\of{t}}\}$ with $\alpha=1.14$, which is the minimax choice, i.e., the value of $\alpha$ that minimizes the maximum MSE over all $0.1$-sparse signals (see \cite{Montanari:Chap:12} for more details).
With the soft-thresholding denoiser, AMP solves the LASSO problem 
``$\argmin_{\vec{x}} \{\frac{1}{2}\|\vec{y}-\vec{Ax}\|^2+\lambda\|\vec{x}\|_1\}$''
for some value of $\lambda$ \cite{Donoho:PNAS:09,Montanari:Chap:12}.

Figures~\ref{fig:mse_fixed1_bg_n300_t10000}-\ref{fig:mse_fixed1_st_n30000_t100} plot finite-dimensional versions of the denoiser output MSE $\mc{E}\of{t}$ and the denoiser input-error variance $\tau_r\of{t}$ versus iteration $t$ for both the AMP algorithm \eqref{alg} and the AMP state evolution \eqref{se}.
For the AMP algorithm, the iteration-$t$ denoiser output MSE was computed as 
$\mc{E}_n\of{t} = \frac{1}{n}\sum_{j=1}^n (x_j-x_j\of{t})^2$ 
and the denoiser input-error variance was computed as 
$\tau_{r,n}\of{t} = \|\vec{v}\of{t}\|^2/m$,
where the subscript $n$ indicates the dimensional dependence of these quantities.
For the AMP state evolution, the denoiser output MSE was computed as 
\begin{align}
\mc{E}_n\of{t} =
\begin{cases}
\E\big[\eta\of{t-1}(X+\mc{N}(0,\tau_{r,n}\of{t-1})) - X\big]^2 & t>0 \\
\E\big[X^2\big] & t=0 ,
\end{cases}
\end{align}
with the expectation evaluated using the $n$-term empirical distribution for $X$, and the iteration-$t$ denoiser input-error variance was computed as $\tau_{r,n}\of{t}=\delta^{-1}\mc{E}_n\of{t}+\tau_{w,n}$ using the empirical noise variance $\tau_{w,n}=\frac{1}{m}\sum_{i=1}^m w_i^2$.
Each figure plots the empirical mean and standard deviation over $T\in\{100,1000,10000\}$ random draws of $\Amat$ for a \emph{single fixed} draw of $\vec{x}$ and $\vec{w}$.

\Figref{mse_fixed1_bg_n300_t10000} shows the results for the MMSE denoiser at dimension $n=300$.
The figure shows a good, but not great, agreement between the state evolution and average AMP quantities, where the average was computed over $T=10000$ realizations of $\Amat$.
Furthermore, the error bars in \figref{mse_fixed1_bg_n300_t10000}, which show the empirical standard deviation over the $T$ realizations, indicate that there was considerable dependence of the trajectories $\{\mc{E}_n\of{t}\}_{t=1}^{30}$ and $\{\tau_{r,n}\of{t}\}_{t=1}^{30}$ on the realization of $\Amat$ when $n=300$.
\Figref{mse_fixed1_bg_n3000_t1000} and~\ref{fig:mse_fixed1_bg_n30000_t100} plot the same quantities for dimensions $n=3000$ and $n=30000$, respectively.
These figures show that, as the dimension $n$ increases, the agreement between the state-evolution and average AMP trajectories improves and the standard deviation of the AMP trajectories decreases. 
\tabref{std} suggests that the standard deviation decreases proportional to $1/\sqrt{n}$.
\Figref{mse_fixed1_bg_n30000_t100} shows that, when $n=30000$, the trajectories $\{\mc{E}_n\of{t}\}_{t=1}^{30}$ and $\{\tau_{r,n}\of{t}\}_{t=1}^{30}$ are nearly invariant to changes in $\Amat$.

Figures~\ref{fig:mse_fixed1_st_n300_t10000}-\ref{fig:mse_fixed1_st_n30000_t100} are similar to Figures~\ref{fig:mse_fixed1_bg_n300_t10000}-\ref{fig:mse_fixed1_bg_n30000_t100}, except that they show the results for the soft-thresholding denoiser.
As expected, the use of the soft-thresholding denoiser results in larger MSEs than the MMSE denoiser.
But, otherwise, the trends are the same:
the agreement between the state-evolution and average AMP trajectories increases with the dimension $n$, and the standard deviation of the AMP trajectories decreases proportional to $1/\sqrt{n}$.

Note that, because the state evolution was computed using the empirical distributions of $\{x_j\}_{j=1}^n$ and $\{w_i\}_{i=1}^m$, which change from one figure to the next (e.g., as $n$ and $m$ change), the state evolution trajectories vary across Figures~\ref{fig:mse_fixed1_bg_n300_t10000}-\ref{fig:mse_fixed1_st_n30000_t100}.

Finally, to give evidence that the denoiser input error $\{e_j\of{t}\}$ is approximately Gaussian, we show quantile-quantile (QQ) plots in
Figures~\ref{fig:qq_bg_i5_n3000} and~\ref{fig:qq_st_i5_n3000}
at iteration $t=5$ and dimension $n=3000$ for the MMSE and soft-thresholding denoisers, respectively.
The figures show that the quantiles of $\{e_j\of{t}\}$ are very close to those of a zero-mean Gaussian random variable.
Although not shown here, QQ plots at other iterations $t$ look similar, and the QQ plots become more linear (i.e., $\{e_j\of{t}\}$ looks more Gaussian) as $n$ grows larger.

\putFrag{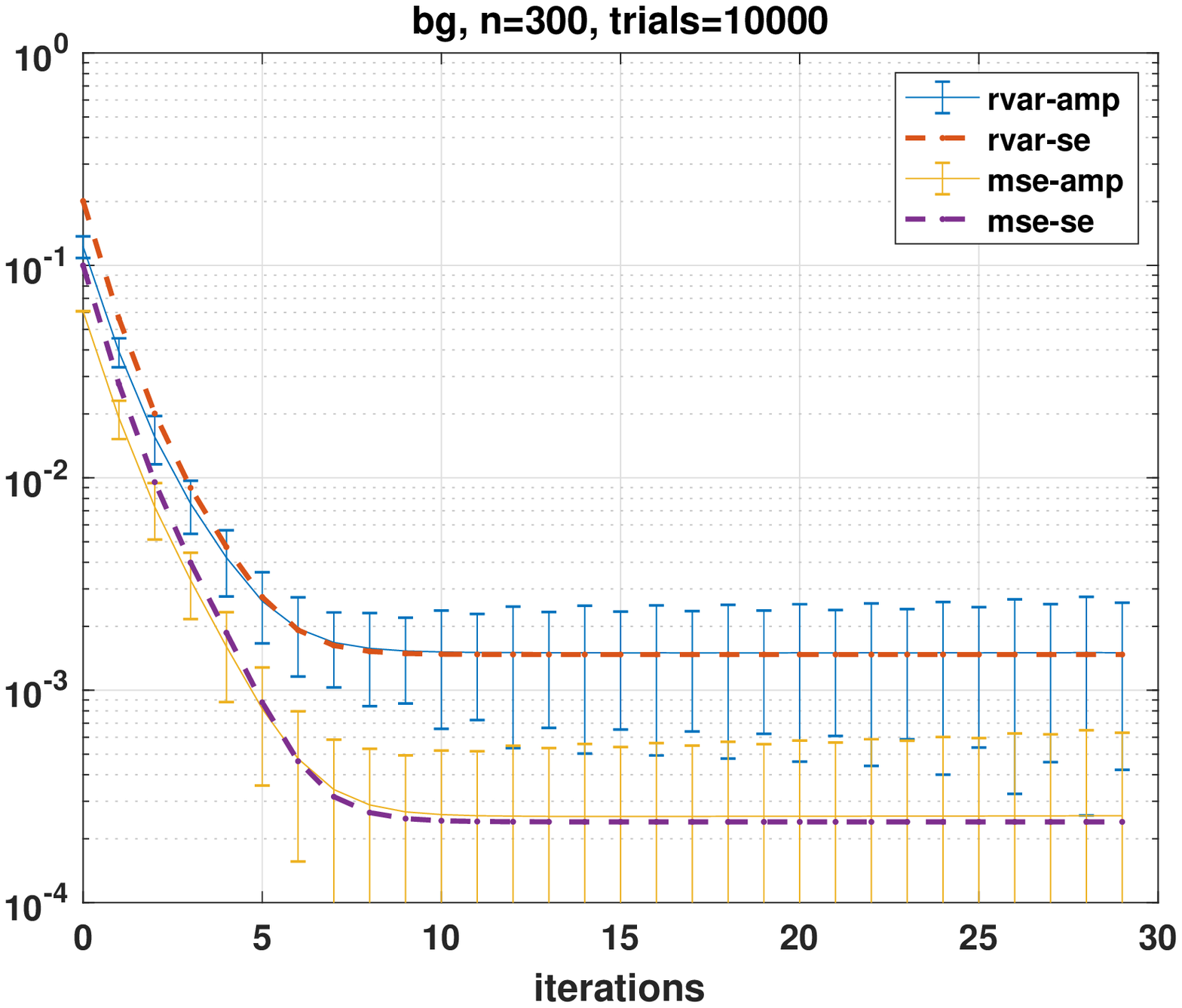}
        {Denoiser output error $\mc{E}_n\of{t}$ and denoiser input-error variance $\tau_{r,n}\of{t}$ versus iteration for AMP and its state evolution with MMSE denoising and $n=300$.  Dashed lines show the empirical average over $10000$ random draws of $\Amat$ and error bars show the empirical standard deviation.}
        {3.0}
        {\psfrag{bg, n=300, trials=10000}{}
         \psfrag{iterations}[t][t]{\sf iteration}
         \psfrag{rvar-amp}[bl][bl][\sz]{\sf $\tau_{r,n}\of{t}$-AMP}
         \psfrag{rvar-se}[bl][bl][\sz]{\sf $\tau_{r,n}\of{t}$-SE}
         \psfrag{mse-amp}[bl][bl][\sz]{\sf $\mc{E}_n\of{t}$-AMP}
         \psfrag{mse-se}[bl][bl][\sz]{\sf $\mc{E}_n\of{t}$-SE}
        }

\putFrag{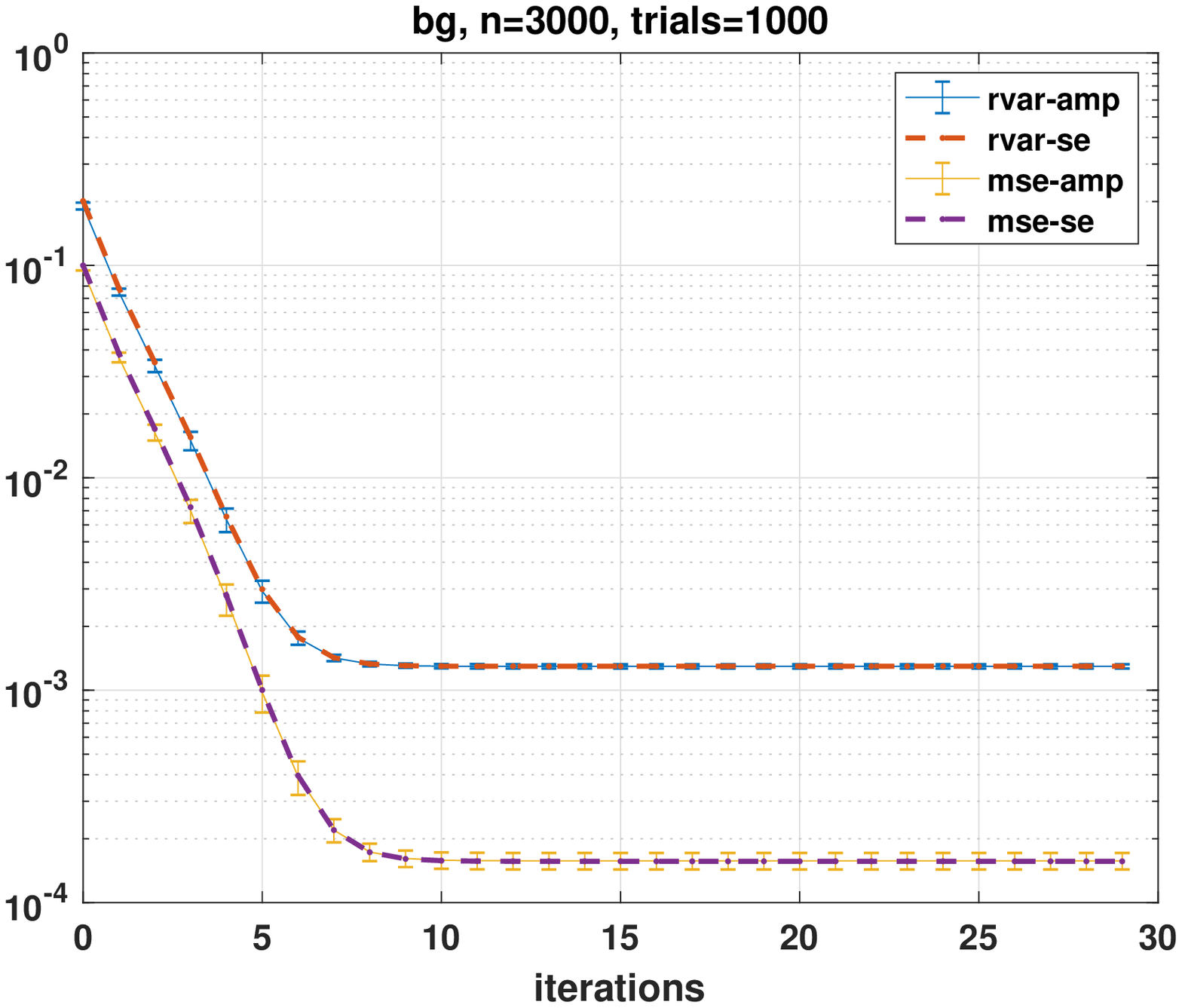}
        {Denoiser output MSE $\mc{E}_n\of{t}$ and denoiser input-error variance $\tau_{r,n}\of{t}$ versus iteration for AMP and its state evolution with MMSE denoising and $n=3000$.  Dashed lines show the empirical average over $1000$ random draws of $\Amat$ and error bars show the empirical standard deviation.}
        {3.0}
        {\psfrag{bg, n=3000, trials=1000}{}
         \psfrag{iterations}[t][t]{\sf iteration}
         \psfrag{rvar-amp}[bl][bl][\sz]{\sf $\tau_{r,n}\of{t}$-AMP}
         \psfrag{rvar-se}[bl][bl][\sz]{\sf $\tau_{r,n}\of{t}$-SE}
         \psfrag{mse-amp}[bl][bl][\sz]{\sf $\mc{E}_n\of{t}$-AMP}
         \psfrag{mse-se}[bl][bl][\sz]{\sf $\mc{E}_n\of{t}$-SE}
        }

\putFrag{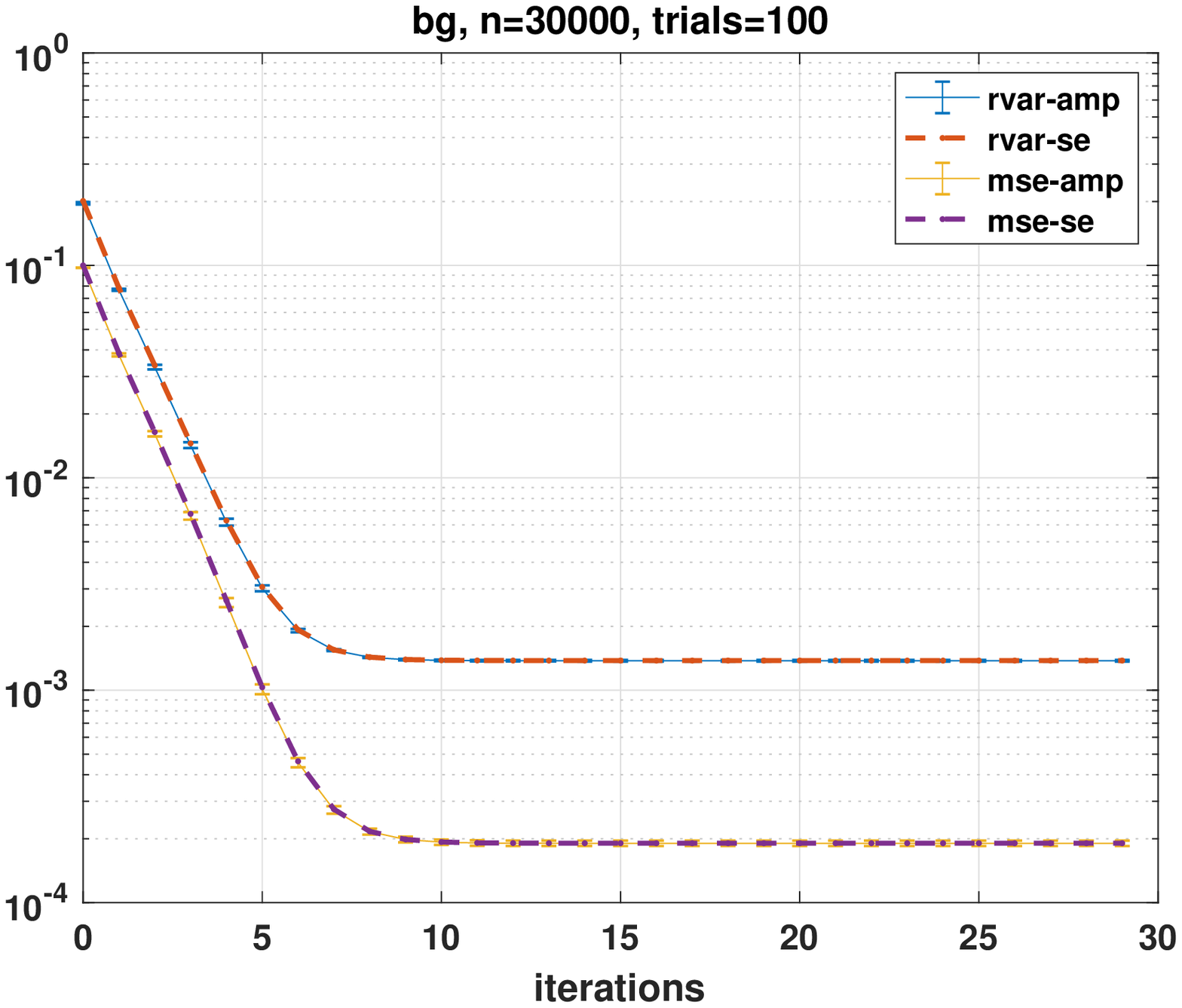}
        {Denoiser output MSE $\mc{E}_n\of{t}$ and denoiser input-error variance $\tau_{r,n}\of{t}$ versus iteration for AMP and its state evolution with MMSE denoising and $n=30000$.  Dashed lines show the empirical average over $100$ random draws of $\Amat$ and error bars show the empirical standard deviation.}
        {3.0}
        {\psfrag{bg, n=30000, trials=100}{}
         \psfrag{iterations}[t][t]{\sf iteration}
         \psfrag{rvar-amp}[bl][bl][\sz]{\sf $\tau_{r,n}\of{t}$-AMP}
         \psfrag{rvar-se}[bl][bl][\sz]{\sf $\tau_{r,n}\of{t}$-SE}
         \psfrag{mse-amp}[bl][bl][\sz]{\sf $\mc{E}_n\of{t}$-AMP}
         \psfrag{mse-se}[bl][bl][\sz]{\sf $\mc{E}_n\of{t}$-SE}
        }

\putFrag{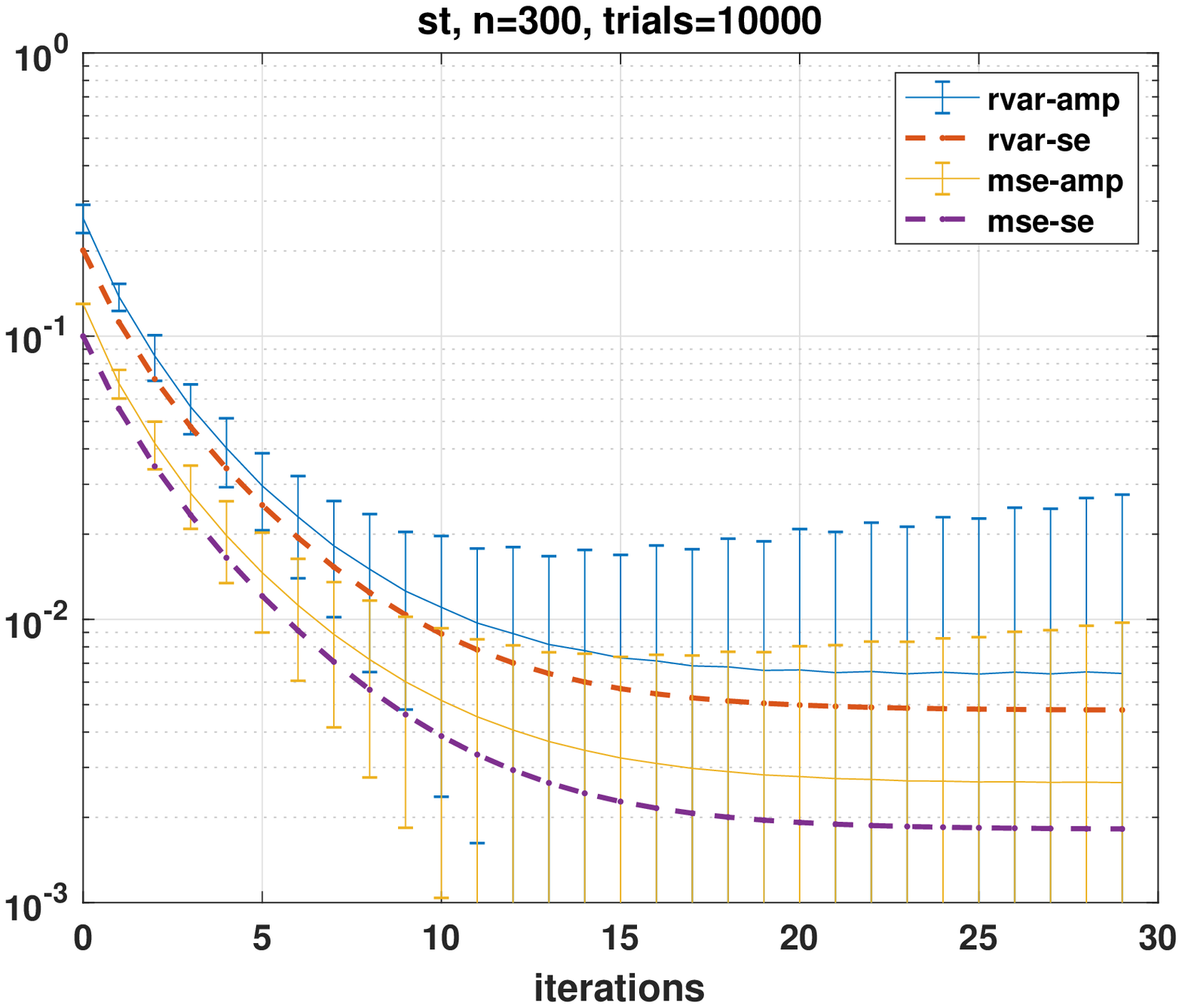}
        {Denoiser output MSE $\mc{E}_n\of{t}$ and denoiser input-error variance $\tau_{r,n}\of{t}$ versus iteration for AMP and its state evolution with soft-threshold denoising and $n=300$.  Dashed lines show the empirical average over $10000$ random draws of $\Amat$ and error bars show the empirical standard deviation.}
        {3.0}
        {\psfrag{st, n=300, trials=10000}{}
         \psfrag{iterations}[t][t]{\sf iteration}
         \psfrag{rvar-amp}[bl][bl][\sz]{\sf $\tau_{r,n}\of{t}$-AMP}
         \psfrag{rvar-se}[bl][bl][\sz]{\sf $\tau_{r,n}\of{t}$-SE}
         \psfrag{mse-amp}[bl][bl][\sz]{\sf $\mc{E}_n\of{t}$-AMP}
         \psfrag{mse-se}[bl][bl][\sz]{\sf $\mc{E}_n\of{t}$-SE}
        }

\putFrag{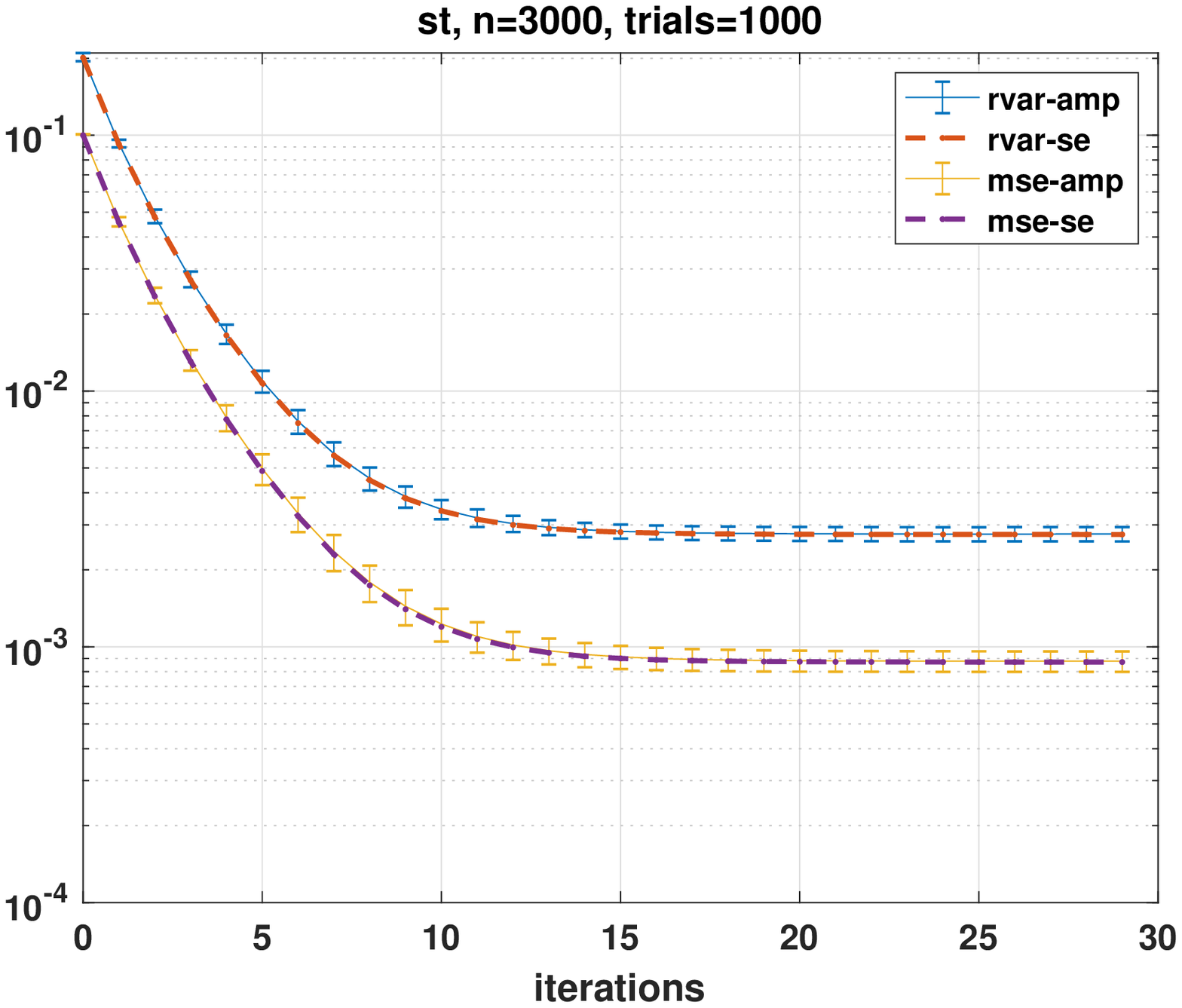}
        {Denoiser output MSE $\mc{E}_n\of{t}$ and denoiser input-error variance $\tau_{r,n}\of{t}$ versus iteration for AMP and its state evolution with soft-threshold denoising and $n=3000$.  Dashed lines show the empirical average over $1000$ random draws of $\Amat$ and error bars show the empirical standard deviation.}
        {3.0}
        {\psfrag{st, n=3000, trials=1000}{}
         \psfrag{iterations}[t][t]{\sf iteration}
         \psfrag{rvar-amp}[bl][bl][\sz]{\sf $\tau_{r,n}\of{t}$-AMP}
         \psfrag{rvar-se}[bl][bl][\sz]{\sf $\tau_{r,n}\of{t}$-SE}
         \psfrag{mse-amp}[bl][bl][\sz]{\sf $\mc{E}_n\of{t}$-AMP}
         \psfrag{mse-se}[bl][bl][\sz]{\sf $\mc{E}_n\of{t}$-SE}
        }

\putFrag{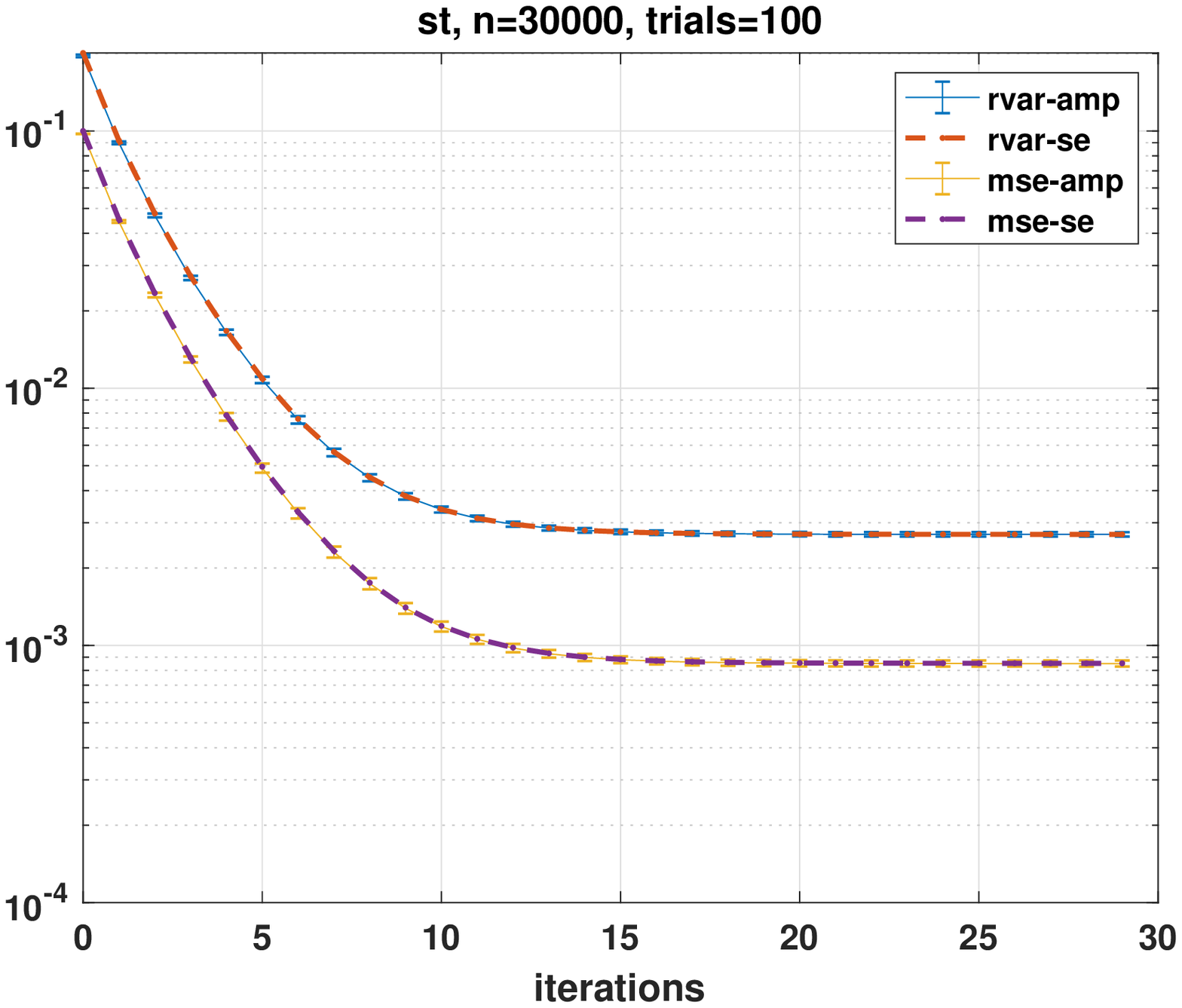}
        {Denoiser output MSE $\mc{E}_n\of{t}$ and denoiser input-error variance $\tau_{r,n}\of{t}$ versus iteration for AMP and its state evolution with soft-threshold denoising and $n=30000$.  Dashed lines show the empirical average over $100$ random draws of $\Amat$ and error bars show the empirical standard deviation.}
        {3.0}
        {\psfrag{st, n=30000, trials=100}{}
         \psfrag{iterations}[t][t]{\sf iteration}
         \psfrag{rvar-amp}[bl][bl][\sz]{\sf $\tau_{r,n}\of{t}$-AMP}
         \psfrag{rvar-se}[bl][bl][\sz]{\sf $\tau_{r,n}\of{t}$-SE}
         \psfrag{mse-amp}[bl][bl][\sz]{\sf $\mc{E}_n\of{t}$-AMP}
         \psfrag{mse-se}[bl][bl][\sz]{\sf $\mc{E}_n\of{t}$-SE}
        }

\putFrag{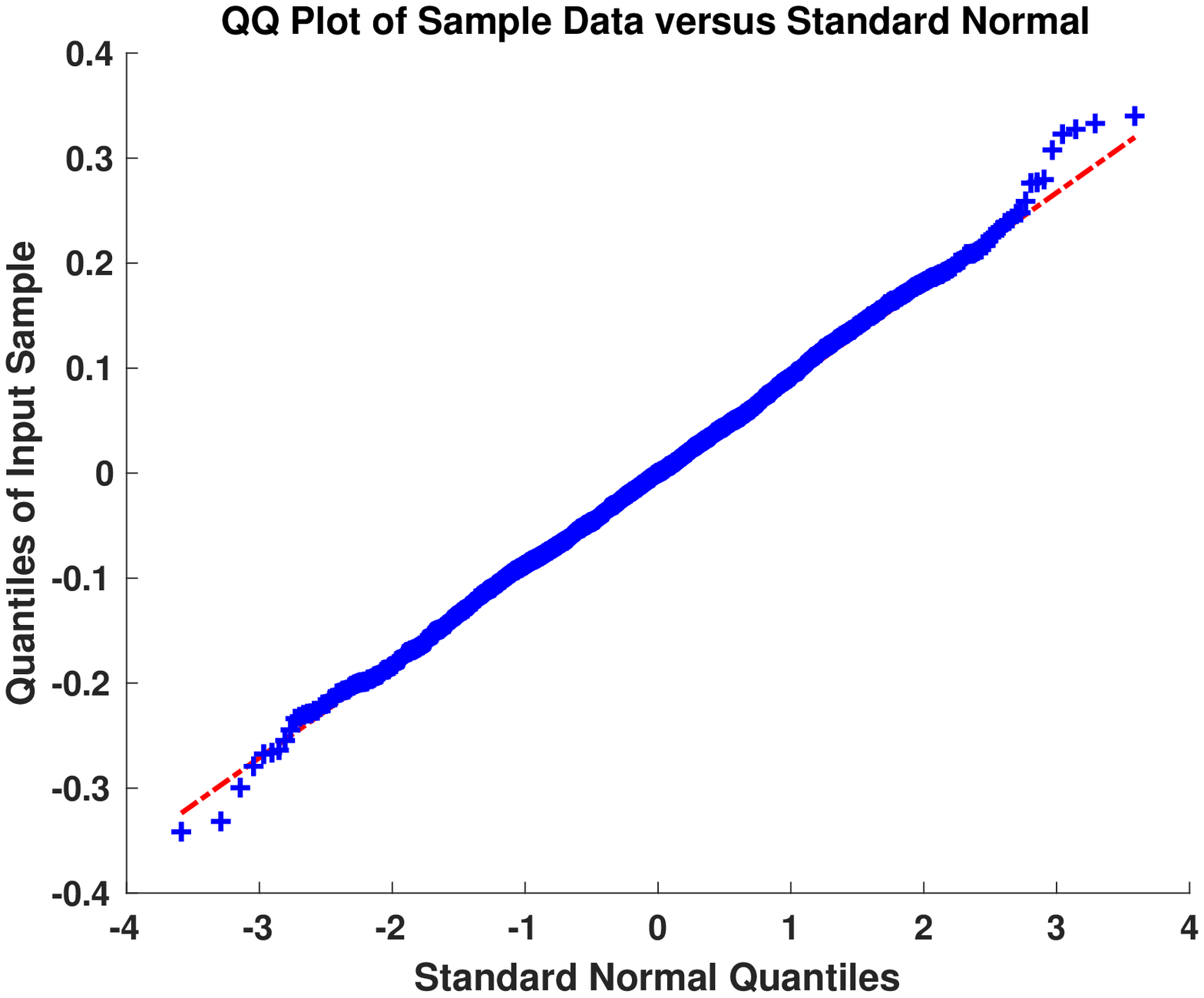}
        {Quantiles of standard normal versus quantiles of 
         AMP denoiser input-error $\{e_j\of{t}\}$ at iteration $t=5$ 
         with MMSE denoising and $n=3000$.}
        {3.0}
        {\psfrag{QQ Plot of Sample Data versus Standard Normal}{}}

\putFrag{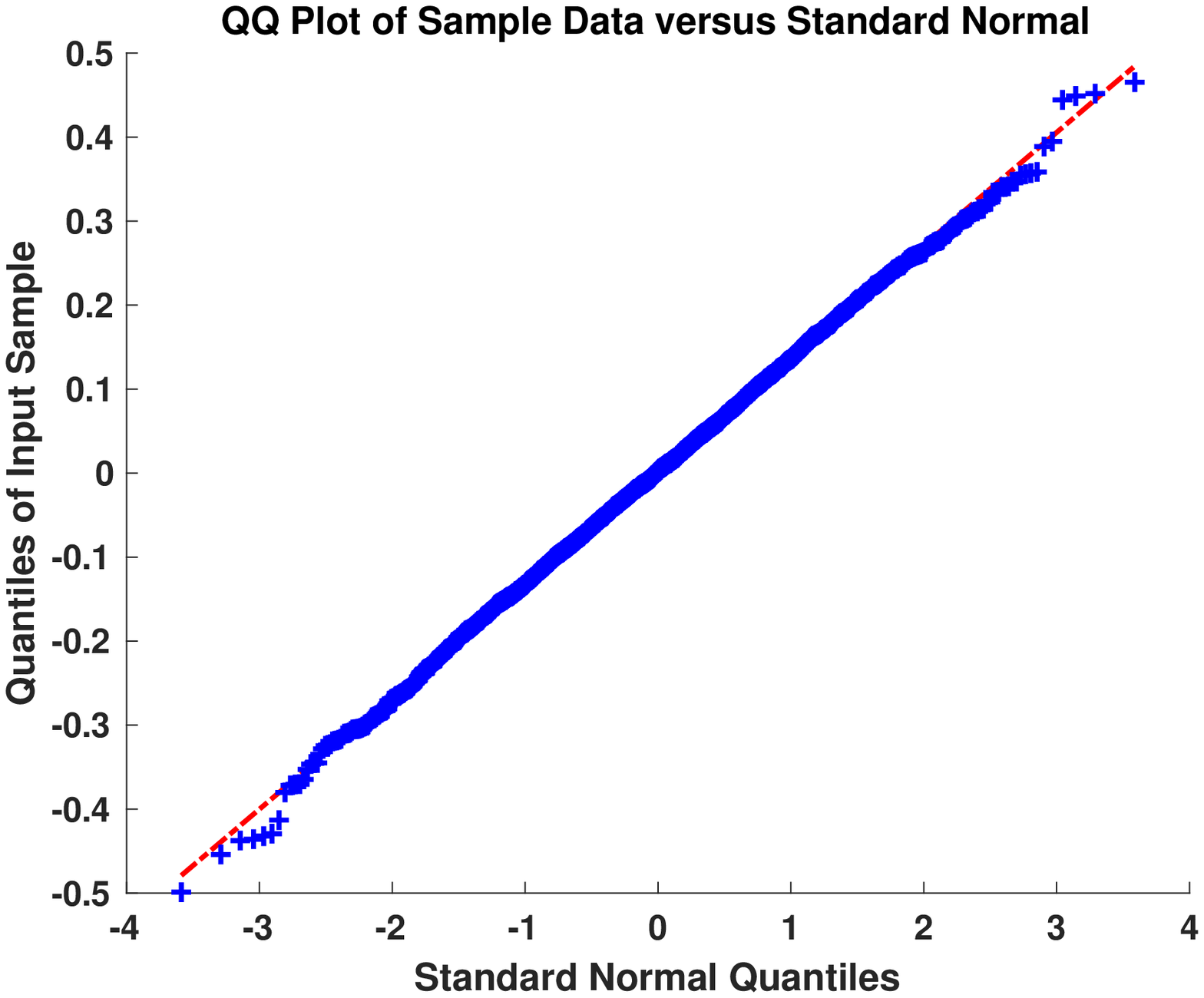}
        {Quantiles of standard normal versus quantiles of 
         AMP denoiser input-error $\{e_j\of{t}\}$ at iteration $t=5$ 
         with soft-threshold denoising and $n=3000$.}
        {3.0}
        {\psfrag{QQ Plot of Sample Data versus Standard Normal}{}}

\begin{table}[t]
 \caption{Numerical evidence that $\text{std}(\mc{E}_n\of{t})\sqrt{n}$ and $\text{std}(\tau_{r,n}\of{t})\sqrt{n}$ are approximately constant with $n$, implying that $\text{std}(\mc{E}_n\of{t})$ and $\text{std}(\tau_{r,n}\of{t})$ scale as approximately $1/\sqrt{n}$.}
 \begin{center}
  \ifarxiv{
  \begin{tabular}{|c|c|cccc|}\hline
   & $n$ & 1000 & 3000 & 10000 & 30000 \\\hline
   \multirow{2}{*}{MMSE denoiser} & $\text{std}(\mc{E}_n\of{29})\sqrt{n}$ 
        & 0.0011 & 0.0008 & 0.0010 & 0.0010 \\
   & $\text{std}(\tau_{r,n}\of{29})\sqrt{n}$ 
        & 0.0026 & 0.0017 & 0.0022 & 0.0020 \\\hline
   \multirow{2}{*}{Soft Threshold} & $\text{std}(\mc{E}_n\of{29})\sqrt{n}$ 
        & 0.0055 & 0.0044 & 0.0037 & 0.0040 \\
   & $\text{std}(\tau_{r,n}\of{29})\sqrt{n}$ 
        & 0.0118 & 0.0098 & 0.0106 & 0.0093 \\\hline
  \end{tabular}
  }\else{
  \begin{tabular}{|@{\;}c@{\;}|@{\;}c@{\;}|cccc|}\hline
   & $n$ & 1000 & 3000 & 10000 & 30000 \\\hline
   \multirow{2}{*}{MMSE denoiser} & $\text{std}(\mc{E}_n\of{29})\sqrt{n}$ 
        & 0.0011 & 0.0008 & 0.0010 & 0.0010 \\
   & $\text{std}(\tau_{r,n}\of{29})\sqrt{n}$ 
        & 0.0026 & 0.0017 & 0.0022 & 0.0020 \\\hline
   \multirow{2}{*}{Soft Threshold} & $\text{std}(\mc{E}_n\of{29})\sqrt{n}$ 
        & 0.0055 & 0.0044 & 0.0037 & 0.0040 \\
   & $\text{std}(\tau_{r,n}\of{29})\sqrt{n}$ 
        & 0.0118 & 0.0098 & 0.0106 & 0.0093 \\\hline
  \end{tabular}
  }\fi
 \end{center}
 \label{tab:std}
\end{table}


\section{Conclusion}

For the linear regression problem, we presented a simple derivation of AMP and its MSE state-evolution based on the idea of ``first-order cancellation.''
In particular, our derivation writes the linear transform of the denoiser output, $\Amat\vec{x}\of{t}$, as the sum of a term that is weakly dependent on the previous iteration and another term that is strongly dependent but canceled by AMP's Onsager correction term in the large-system limit.
Our derivation provides insights that are missing from the usual loopy belief-propagation derivation of AMP, while being much more accessible than Bayati et al.'s rigorous analysis of AMP. 


\appendix

Here we establish \lemref{scaling2}, which says that the elements of $\vec{v}\of{t}, \vec{r}\of{t}, \vec{x}\of{t}$, and $\vec{\mu}\of{t}$ scale as $O(1)$ in the large-system limit under \assref{indep}.
We do this by induction.

From the initialization $\vec{x}\of{0}=\vec{0}=\vec{\mu}\of{0}$, we have that $\vec{v}\of{0}=\vec{y}$, whose elements are $O(1)$, and we have 
$\vec{r}\of{0}=\vec{A}\tran\vec{y}=\vec{A}\tran\vec{Ax}+\vec{A}\tran\vec{w}$.
Examining the $j$th entry, we see that
\begin{align}
r_j\of{0}
&= \sum_{i} a_{ij} \sum_{l} a_{il} x_l + \sum_{i} a_{ij} w_i \\
&= x_j \underbrace{ \sum_{i=1}^m a_{ij}^2 }_{\displaystyle = 1} + \underbrace{ \sum_{i} a_{ij} \underbrace{\sum_{l\neq j} a_{il} x_l}_{\displaystyle O(1)} }_{\displaystyle O(1)} + \underbrace{\sum_{i} a_{ij} w_i}_{\displaystyle O(1)}
\end{align}
since $a_{ij}^2=1/m~\forall ij$ and where the $O(1)$ scalings follow from \textb{\lemref{scaling}} under \textb{\assref{indep}.}
Thus the elements of $\vec{r}\of{0}$ are $O(1)$.
Because $x_j\of{1}=\denoiser\of{0}(r_j\of{0})$, the elements of $\vec{x}\of{1}$ are also $O(1)$.
And from \eqref{onsager},
\begin{align}
\mu_i\of{1}
&= \frac{n}{m} v_i\of{0} \frac{1}{n}\sum_{j=1}^n \denoiser\ofp{0}(r_j\of{0}) ,
\end{align}
where $\frac{n}{m}$, $v_i\of{0}$, and $\frac{1}{n}\sum_{j=1}^n \denoiser\ofp{0}(r_j\of{0})$ are all $O(1)$, implying that the elements of $\vec{\mu}\of{1}$ are $O(1)$.

Now, suppose that the elements of $\vec{r}\of{t-1},\vec{v}\of{t-1},\vec{x}\of{t},\vec{\mu}\of{t}$ are all $O(1)$, which we know occurs when $t=1$.  
Then from \eqref{alg} we have that
\ifarxiv{
\begin{align}
v_i\of{t}
&= y_i - \sum_{l=1}^n a_{il} x_l\of{t} + \mu_i\of{t} \\
&= y_i - \sum_{l=1}^n a_{il} \denoiser\of{t-1}(r_l\of{t-1}) + \mu_i\of{t} \\
&= y_i - \sum_{l=1}^n a_{il} \denoiser\of{t-1}\Big(
\underbrace{ x_l\of{t-1}+\sum_{k\neq i}a_{kl}v_k\of{t-1} 
        }_{\displaystyle = r_{il}\of{t-1}}
+a_{il}v_i\of{t-1}\Big) + \mu_i\of{t} ,
\end{align}
}\else{
\begin{align}
v_i\of{t}
&= y_i - \sum_{l=1}^n a_{il} x_l\of{t} + \mu_i\of{t} \\
&= y_i - \sum_{l=1}^n a_{il} \denoiser\of{t-1}(r_l\of{t-1}) + \mu_i\of{t} \\
&= y_i - \sum_{l=1}^n a_{il} \denoiser\of{t-1}\Big(
\underbrace{ x_l\of{t-1}+\sum_{k\neq i}a_{kl}v_k\of{t-1} 
        }_{\displaystyle = r_{il}\of{t-1}}
+a_{il}v_i\of{t-1}\Big) 
\nonumber\\&\quad 
+ \mu_i\of{t} ,
\end{align}
}\fi
where $r_{il}\of{t-1}$ is only weakly dependent on $\{a_{ij}\}_{j=1}^n$, leading us to invoke \assref{indep}.
The Taylor expansion \eqref{Taylor1} (which is admissible since, under the induction hypothesis, the elements of $\vec{v}\of{t-1}$ and $\vec{r}\of{t-1}$ scale as $O(1)$) then yields
\ifarxiv{
\begin{align}
v_i\of{t}
&= y_i - \sum_{l=1}^n a_{il} \big[ \denoiser\of{t-1}(r_{il}\of{t-1}) + a_{il}v_i\of{t-1} \denoiser\ofp{t-1}(r_{il}\of{t-1}) + O(1/m) \big] + \mu_i\of{t} \\
&= y_i - \underbrace{ \sum_{l=1}^n a_{il} \denoiser\of{t-1}(r_{il}\of{t-1}) }_{\displaystyle O(1) } 
- v_i\of{t-1} \underbrace{\frac{1}{m}\sum_{l=1}^n \denoiser\ofp{t-1}(r_{il}\of{t-1})}_{\displaystyle O(1)} 
+ O(1/m)\underbrace{\sum_{l=1}^n a_{il}}_{\displaystyle O(1)} 
+ \mu_i\of{t} 
\label{eq:vit},
\end{align}
}\else{
\begin{align}
v_i\of{t}
&= y_i - \sum_{l=1}^n a_{il} \Big[ \denoiser\of{t-1}(r_{il}\of{t-1}) + a_{il}v_i\of{t-1} \denoiser\ofp{t-1}(r_{il}\of{t-1}) 
\nonumber\\&\quad
+ O(1/m) \Big] + \mu_i\of{t} \\
&= y_i - \underbrace{ \sum_{l=1}^n a_{il} \denoiser\of{t-1}(r_{il}\of{t-1}) }_{\displaystyle O(1) } 
- v_i\of{t-1} \underbrace{\frac{1}{m}\sum_{l=1}^n \denoiser\ofp{t-1}(r_{il}\of{t-1})}_{\displaystyle O(1)} 
\nonumber\\&\quad
+ O(1/m)\underbrace{\sum_{l=1}^n a_{il}}_{\displaystyle O(1)} 
+ \mu_i\of{t} 
\label{eq:vit},
\end{align}
}\fi
where we used the fact that $a_{il}^2=1/m~\forall il$.
In \eqref{vit}, the $O(1)$ scaling of the first and last sums follows from \textb{\lemref{scaling}} under \assref{indep}.
Thus the elements of $\vec{v}\of{t}$ are $O(1)$.

Next, we establish that the elements of $\vec{r}\of{t}$ are $O(1)$ 
when the elements of $\vec{v}\of{t-1},\vec{x}\of{t}$ and $\vec{\mu}\of{t}$ are.
From \eqref{et0}-\eqref{et} we have
\ifarxiv{
\begin{align}
\vec{r}\of{t} 
&= (\vec{I}-\Amat\tran\Amat)\vec{x}\of{t} - (\vec{I}-\Amat\tran\Amat)\vec{x} + \Amat\tran(\vec{w}+\vec{\mu}\of{t}) + \vec{x} \\
&= (\vec{I}-\Amat\tran\Amat)\vec{x}\of{t} + \vec{A}\tran\vec{\mu}\of{t} 
+ \underbrace{ \Amat\tran\Amat\vec{x} + \Amat\tran\vec{w} }_{\displaystyle =\vec{r}\of{0}} 
\label{eq:rt},
\end{align}
}\else{
\begin{align}
\vec{r}\of{t} 
&= (\vec{I}-\Amat\tran\Amat)\vec{x}\of{t} - (\vec{I}-\Amat\tran\Amat)\vec{x} + \Amat\tran(\vec{w}+\vec{\mu}\of{t}) 
\nonumber\\&\quad 
+ \vec{x} \\
&= (\vec{I}-\Amat\tran\Amat)\vec{x}\of{t} + \vec{A}\tran\vec{\mu}\of{t} 
+ \underbrace{ \Amat\tran\Amat\vec{x} + \Amat\tran\vec{w} }_{\displaystyle =\vec{r}\of{0}} 
\label{eq:rt},
\end{align}
}\fi
where we previously established that the elements in $\vec{r}\of{0}$ are $O(1)$.
As for the remaining term in \eqref{rt}, we have from \eqref{IAAxtj} and \eqref{onsager} that
\ifarxiv{
\begin{align}
\lefteqn{ \big[ (\vec{I}-\Amat\tran\Amat)\vec{x}\of{t} + \vec{A}\tran\vec{\mu}\of{t} \big]_j }\nonumber\\
&= -\sum_i a_{ij} \sum_{l\neq j} a_{il} \denoiser\of{t-1}\big( r_{il}\of{t-1} + a_{il}v_i\of{t-1}\big) + \sum_i a_{ij} \frac{1}{m} v_i\of{t-1}  \sum_l \denoiser\ofp{t-1}\big(r_l\of{t-1}\big) \\
&= \sum_i a_{ij} \Big[ \frac{1}{m} v_i\of{t-1}  \sum_l \denoiser\ofp{t-1}\big(r_l\of{t-1}\big) 
-\sum_{l\neq j} a_{il} \denoiser\of{t-1}\big( r_{il}\of{t-1} + a_{il}v_i\of{t-1}\big) \Big] \\
&= \sum_i a_{ij} \Big[ \frac{1}{m} v_i\of{t-1}  \sum_l \denoiser\ofp{t-1}\big(r_l\of{t-1}\big) 
-\sum_{l\neq j} a_{il} \denoiser\of{t-1}\big( r_{il}\of{t-1} \big) 
-\sum_{l\neq j} a_{il}^2 v_i\of{t-1} \denoiser\ofp{t-1}\big(r_{il}\of{t-1}\big) 
+ \underbrace{ O(1/m)\sum_{l\neq j}a_{il} }_{\displaystyle O(1/\sqrt{m})} \Big] ,
\end{align}
}\else{
\begin{align}
\lefteqn{ \big[ (\vec{I}-\Amat\tran\Amat)\vec{x}\of{t} + \vec{A}\tran\vec{\mu}\of{t} \big]_j }\nonumber\\
&= -\sum_i a_{ij} \sum_{l\neq j} a_{il} \denoiser\of{t-1}\big( r_{il}\of{t-1} + a_{il}v_i\of{t-1}\big) 
\nonumber\\&\quad
+ \sum_i a_{ij} \frac{1}{m} v_i\of{t-1}  \sum_l \denoiser\ofp{t-1}\big(r_l\of{t-1}\big) \\
&= \sum_i a_{ij} \Big[ \frac{1}{m} v_i\of{t-1}  \sum_l \denoiser\ofp{t-1}\big(r_l\of{t-1}\big) 
\nonumber\\&\quad
-\sum_{l\neq j} a_{il} \denoiser\of{t-1}\big( r_{il}\of{t-1} + a_{il}v_i\of{t-1}\big) \Big] \\
&= \sum_i a_{ij} \Big[ \frac{1}{m} v_i\of{t-1}  \sum_l \denoiser\ofp{t-1}\big(r_l\of{t-1}\big) 
-\sum_{l\neq j} a_{il} \denoiser\of{t-1}\big( r_{il}\of{t-1} \big) 
\nonumber\\&\quad
-\sum_{l\neq j} a_{il}^2 v_i\of{t-1} \denoiser\ofp{t-1}\big(r_{il}\of{t-1}\big) 
+ \underbrace{ O(1/m)\sum_{l\neq j}a_{il} }_{\displaystyle O(1/\sqrt{m})} \Big] ,
\end{align}
}\fi
where, for the last step, we applied the Taylor expansion \eqref{Taylor1}. 
Applying $a_{il}^2=1/m~\forall il$ and rearranging, we get
\ifarxiv{
\begin{align}
\lefteqn{ \big[ (\vec{I}-\Amat\tran\Amat)\vec{x}\of{t} + \vec{A}\tran\vec{\mu}\of{t} \big]_j }\nonumber\\
&= \sum_i a_{ij} \Big[ \frac{1}{m} v_i\of{t-1}  \sum_l \denoiser\ofp{t-1}\big(r_l\of{t-1}\big) 
-\frac{1}{m} v_i\of{t-1} \sum_{l\neq j} \denoiser\ofp{t-1}\big(r_{il}\of{t-1}\big) 
-\sum_{l\neq j} a_{il} \denoiser\of{t-1}\big( r_{il}\of{t-1} \big) 
+ O(1/\sqrt{m}) \Big] \\
&= \frac{1}{m}\sum_i a_{ij} v_i\of{t-1} \Big[ \sum_l \denoiser\ofp{t-1}\big(r_l\of{t-1}\big)
- \sum_{l\neq j} \denoiser\ofp{t-1}\big(r_{il}\of{t-1}\big) \Big]
-\sum_i a_{ij} \sum_{l\neq j} a_{il} \denoiser\of{t-1}\big( r_{il}\of{t-1} \big)
+ O(1) \\
&= \frac{1}{m}\sum_i a_{ij} v_i\of{t-1} \Big[ \denoiser\ofp{t-1}\big(r_j\of{t-1}\big) 
+ \sum_{l\neq j} \Big( \denoiser\ofp{t-1}\big(r_l\of{t-1}\big)
- \denoiser\ofp{t-1}\big(r_{il}\of{t-1}\big) \Big) \Big]
-\sum_i a_{ij} \sum_{l\neq j} a_{il} \denoiser\of{t-1}\big( r_{il}\of{t-1} \big)
+ O(1) 
\label{eq:pain},
\end{align}
}\else{
\begin{align}
\lefteqn{ \big[ (\vec{I}-\Amat\tran\Amat)\vec{x}\of{t} + \vec{A}\tran\vec{\mu}\of{t} \big]_j }\nonumber\\
&= \sum_i a_{ij} \Big[ \frac{v_i\of{t-1}}{m}  \sum_l \denoiser\ofp{t-1}\big(r_l\of{t-1}\big) 
-\frac{v_i\of{t-1}}{m} \sum_{l\neq j} \denoiser\ofp{t-1}\big(r_{il}\of{t-1}\big)
\nonumber\\&\quad
-\sum_{l\neq j} a_{il} \denoiser\of{t-1}\big( r_{il}\of{t-1} \big) 
+ O(1/\sqrt{m}) \Big] \\
&= \frac{1}{m}\sum_i a_{ij} v_i\of{t-1} \Big[ \sum_l \denoiser\ofp{t-1}\big(r_l\of{t-1}\big)
- \sum_{l\neq j} \denoiser\ofp{t-1}\big(r_{il}\of{t-1}\big) \Big]
\nonumber\\&\quad
-\sum_i a_{ij} \sum_{l\neq j} a_{il} \denoiser\of{t-1}\big( r_{il}\of{t-1} \big)
+ O(1) \\
&= \frac{1}{m}\sum_i a_{ij} v_i\of{t-1} \Big[ \denoiser\ofp{t-1}\big(r_j\of{t-1}\big) 
\nonumber\\&\quad
+ \sum_{l\neq j} \Big( \denoiser\ofp{t-1}\big(r_l\of{t-1}\big)
- \denoiser\ofp{t-1}\big(r_{il}\of{t-1}\big) \Big) \Big]
\nonumber\\&\quad
-\sum_i a_{ij} \sum_{l\neq j} a_{il} \denoiser\of{t-1}\big( r_{il}\of{t-1} \big)
+ O(1) 
\label{eq:pain},
\end{align}
}\fi
where the $O(1)$ term follows from the fact that $a_{ij}\in\pm 1/\sqrt{m}$.
For the 2nd-to-last term in \eqref{pain}, we have
\begin{align}
\sum_i a_{ij} \underbrace{ \sum_{l\neq j} a_{il} \denoiser\of{t-1}\big( r_{il}\of{t-1} \big) }_{\displaystyle O(1)} = O(1) ,
\end{align}
which follows from \textb{\lemref{scaling}} under \assref{indep}.
For the remaining term in \eqref{pain}, we apply the Taylor expansion \eqref{Taylor2} and rearrange terms as follows: 
\ifarxiv{
\begin{align}
\lefteqn{
\frac{1}{m}\sum_i a_{ij} v_i\of{t-1} \Big[ \denoiser\ofp{t-1}\big(r_j\of{t-1}\big) 
+ \sum_{l\neq j} \Big( \denoiser\ofp{t-1}\big(r_l\of{t-1}\big)
- \denoiser\ofp{t-1}\big(r_{il}\of{t-1}\big) \Big) \Big] }\nonumber\\
&= \frac{1}{m}\sum_i a_{ij} v_i\of{t-1} \Big[ \denoiser\ofp{t-1}\big(r_j\of{t-1}\big) 
+ \sum_{l\neq j} \Big( 
a_{il}v_i\of{t-1}\denoiser\ofpp{t-1}(r_{il}\of{t-1}) + O(1/m) 
\Big) \Big] \\
&= \frac{1}{m}\sum_{i=1}^m a_{ij} v_i\of{t-1} \denoiser\ofp{t-1}\big(r_j\of{t-1}\big) 
+ \frac{1}{m}\sum_{i=1}^m a_{ij} (v_i\of{t-1})^2
\underbrace{ \sum_{l\neq j} a_{il}\denoiser\ofpp{t-1}(r_{il}\of{t-1}) }_{\displaystyle O(1)}
+ \frac{1}{m}\sum_{i=1}^m a_{ij} (v_i\of{t-1})^2 \sum_{l\neq j} O(1/m) 
\label{eq:pain2} \\
&= O(1/\sqrt{m})
\label{eq:pain3} 
\end{align}
}\else{
\begin{align}
\lefteqn{
\frac{1}{m}\sum_i a_{ij} v_i\of{t-1} \Big[ \denoiser\ofp{t-1}\big(r_j\of{t-1}\big) }\nonumber\\
\lefteqn{\quad
+ \sum_{l\neq j} \Big( \denoiser\ofp{t-1}\big(r_l\of{t-1}\big)
- \denoiser\ofp{t-1}\big(r_{il}\of{t-1}\big) \Big) \Big] }\nonumber\\
&= \frac{1}{m}\sum_i a_{ij} v_i\of{t-1} \Big[ \denoiser\ofp{t-1}\big(r_j\of{t-1}\big) 
\nonumber\\&\quad
+ \sum_{l\neq j} \Big( 
a_{il}v_i\of{t-1}\denoiser\ofpp{t-1}(r_{il}\of{t-1}) + O(1/m) 
\Big) \Big] \\
&= \frac{1}{m}\sum_{i=1}^m a_{ij} v_i\of{t-1} \denoiser\ofp{t-1}\big(r_j\of{t-1}\big) 
\nonumber\\&\quad
+ \frac{1}{m}\sum_{i=1}^m a_{ij} (v_i\of{t-1})^2
\underbrace{ \sum_{l\neq j} a_{il}\denoiser\ofpp{t-1}(r_{il}\of{t-1}) }_{\displaystyle O(1)}
\nonumber\\&\quad
+ \frac{1}{m}\sum_{i=1}^m a_{ij} (v_i\of{t-1})^2 \sum_{l\neq j} O(1/m) 
\label{eq:pain2} \\
&= O(1/\sqrt{m})
\label{eq:pain3} 
\end{align}
}\fi
where the $O(1)$ scaling of the sum follows from \textb{\lemref{scaling}} under \assref{indep}.
The final $O(1/\sqrt{m})$ scaling in \eqref{pain3} follows since each of the three terms in \eqref{pain2} is the average of $O(1/\sqrt{m})$ terms, due to $a_{ij}\in\pm 1/\sqrt{m}$.
In conclusion, we have established that the elements of $\vec{r}\of{t}$ are $O(1)$ 
when the elements of $\vec{v}\of{t-1},\vec{x}\of{t}$ and $\vec{\mu}\of{t}$ are.

To complete the induction proof, we need to establish that the elements of $\vec{x}\of{t+1}$ and $\vec{\mu}\of{t+1}$ are $O(1)$ when those of $\vec{v}\of{t}$ and $\vec{r}\of{t}$ are.
But this follows straightforwardly from \eqref{alg} and \eqref{onsager}, i.e.,
\begin{align}
x_j\of{t+1} &= \denoiser\of{t}(r_j\of{t}) \\
\mu_i\of{t+1}
&= \frac{n}{m} v_i\of{t} \frac{1}{n}\sum_{j=1}^n \denoiser\ofp{t}(r_j\of{t}) ,
\end{align}
since $n/m$ is $O(1)$ in the large-system limit.

In summary, we have established \lemref{scaling2}, which says the elements of $\vec{v}\of{t}, \vec{r}\of{t}, \vec{x}\of{t}$, and $\vec{\mu}\of{t}$ scale as $O(1)$ in the large-system limit under \assref{indep}.


\ifarxiv{}\else{
\section*{Acknowledgments}
The author thanks Galen Reeves for inspiring discussions and feedback on early drafts of this manuscript.
}\fi

\bibliographystyle{ieeetr}
\bibliography{macros,stc,books,blind,comm,misc,multicarrier,sparse,machine}

\end{document}